\documentclass[notitlepage,nofootinbib,superscriptaddress]{revtex4-2}
\usepackage[utf8]{inputenc}
\usepackage[T1]{fontenc}

\usepackage{graphicx,graphics,epsfig,subfigure,times,bm,bbm,amssymb,amsmath,amsfonts,amsthm,mathrsfs}

\usepackage[pdfstartview=FitH]{hyperref}
\usepackage{hypernat,ulem}

\usepackage{braket}

\newcommand{\beq}{\begin{equation}}
\newcommand{\eeq}{\end{equation}}
\newcommand{\bes}{\begin{subequations}}
\newcommand{\ees}{\end{subequations}}

\newcommand{\Had}{H_{\text{ad}}}
\newcommand{\Uad}{U_{\text{ad}}}
\newcommand{\Utot}{U_{\text{tot}}}
\newcommand{\D}{\Delta}
\newcommand{\mc}[1]{\mathcal{#1}}

\newtheorem{mydef}{Definition}

\newtheorem{mytheorem}{Theorem}
\newtheorem{mylemma}{Lemma}
\newtheorem{mycorollary}{Corollary}

\newtheorem{myconjecture}{Conjecture}

\begin{document}

\title{Quantum adiabatic theorem for unbounded Hamiltonians with a cutoff and its application to superconducting circuits}
\author{Evgeny Mozgunov$^{1}$, Daniel A. Lidar$^{1-4}$} 
\address{$^{1}$Center for Quantum Information Science \& Technology, University of Southern California, Los Angeles, California 90089, USA\\
$^{2}$Department of Electrical \& Computer Engineering, University of Southern California, Los Angeles, California 90089, USA\\
$^{3}$Department of Physics \& Astronomy, University of Southern California, Los Angeles, California 90089, USA\\
$^{4}$Department of Chemistry, University of Southern California, Los Angeles, California 90089, USA}

\begin{abstract}
We present a new quantum adiabatic theorem that allows one to rigorously bound the adiabatic timescale for a variety of systems, including those described by originally unbounded Hamiltonians that are made finite-dimensional by a cutoff. Our bound is geared towards the qubit approximation of superconducting circuits, and presents a sufficient condition for remaining within the $2^n$-dimensional qubit subspace of a circuit model of $n$ qubits. The novelty of this adiabatic theorem is that unlike previous rigorous results, it does not contain $2^n$ as a factor in the adiabatic timescale, and it allows one to obtain an expression for the adiabatic timescale independent of the cutoff of the infinite-dimensional Hilbert space of the circuit Hamiltonian. As an application, we present an explicit dependence of this timescale on circuit parameters for a superconducting flux qubit, and demonstrate that leakage out of the qubit subspace is inevitable as the tunnelling barrier is raised towards the end of a quantum anneal. We also discuss a method of obtaining a $2^n\times 2^n$ effective Hamiltonian that best approximates the true dynamics induced by slowly changing circuit control parameters.
\end{abstract}

\maketitle


\section{Introduction}
The quantum adiabatic theorem is now more than 100 years old, dating back to Einstein~\cite{Einstein:adiabatic} and Ehrenfest~\cite{Ehrenfest:adiabatic}. Yet, it still continues to inspire new interest and results, in large part owing to its central role in adiabatic quantum computation and quantum annealing, where it can be viewed as providing a sufficient condition for the solution of hard computational problems via adiabatic quantum evolutions~\cite{farhi_quantum_2000,morita:125210,Albash-Lidar:RMP}.

Consider a closed quantum system evolving for a total time $t_f$ subject to the Hamiltonian $H(t)$. Defining the rescaled (dimensionless) time $s=t/t_f$, the evolution is governed by the unitary operator $\Utot(s)$ which is the solution of\footnote{We use a prime to denote $\frac{\partial}{\partial s}$ in this work.}
\begin{equation}
    \Utot'(s) = -it_f H(s) \Utot(s), \quad \Utot(0) = I, \quad s\in[0,1] .
    \label{eq:exact}
\end{equation}

In this work, we assume that the Hamiltonian $H(s) \equiv H_\Lambda(s)$ is a defined as an operator on a finite-dimensional Hilbert space $\mc{H}$ of dimension $\Lambda$, but it is obtained via discretization of an unbounded Hamiltonian $H_\infty$ over an infinite-dimensional Hilbert space. By unbounded we mean that the energy expectation value $\langle \psi | H_\infty |\psi \rangle$ can be arbitrarily large for an appropriate choice of $|\psi\rangle$ within the domain where $H_\infty$ is defined. We will not, however, work with that unbounded Hamiltonian directly, so all our proofs will use the properties of finite-dimensional Hamiltonians, e.g., that the solution to the Schr\"odinger equation exists and the spectrum of $H_\Lambda(s)$ comprises $\Lambda$ discrete (possibly degenerate) eigenvalues. In particular, we will not assume that the limit $\Lambda\to \infty$ of any of the quantities appearing in our results exists. The dimension $\Lambda<\infty$ is what throughout this work we call the {\textit{cutoff}}. We will outline a path to proving a somewhat weaker result for unbounded Hamiltonians $H_\infty$ themselves, but leave a rigorous proof for future work. 

Let $P(s)$ be a finite-rank projection on the low-energy subspace of $H(s)$, i.e., the (continuous-in-$s$) subspace spanned by the eigenvectors with the lowest $d(s)$ eigenvalues. 
A unitary operator $\Uad(s)$ can be constructed that preserves this subspace, i.e.:
\begin{equation}
    P(s) = \Uad (s)P(0) \Uad^\dag (s) .
    \label{eq:intertwining0}
\end{equation}

The adiabatic theorem is essentially the statement that
there exists
$U_{\text{ad}}$ such that
the following holds:\footnote{The norm we use in this work is the operator norm $\|A\| = \sup_{\ket{\psi}}\|A\ket{\psi}\|$ ($\|\ket{\psi}\|=1$),
which is unitarily invariant~\cite{Bhatia:book}: $\|UAV\|=\|A\|$ for arbitrary $A$ and unitary $U$ and $V$.  Additionally, $\|A\| = \|A^\dagger \|$, $\|U\|=1$. Unitarily invariant norms are also submultiplicative:
$\|AB\| \leq \|A\|\|B\|$. For Hermitian operators ($A^\dagger =A$) we have $\|A\| =\sup_{\ket{\psi}}\bra{\psi}\sqrt{A^\dagger A}\ket{\psi} = \sup_{\ket{\psi}}|{\langle \psi| A|\psi\rangle}| \geq {\langle \psi| A|\psi\rangle}$.}  
\begin{equation}
    \|[\Uad (s) -  \Utot(s)]P(0)\|  \leq \frac{\theta}{t_f} \equiv b ,
    \label{eq:diff}
\end{equation}
where $\theta$ is a constant that does not depend on the final time $t_f$, but typically (though not always~\cite{avron_adiabatic_1999,Teufel:book})
depends on the minimum eigenvalue gap $\D$ of $H(s)$ between $P(s)\mc{H}$ and $Q(s)\mc{H}$, where $Q=I-P$. Since the right-hand side (r.h.s.) represents the deviation from adiabaticity, henceforth we refer to $b$ as the 'diabatic evolution bound' and to $\theta$ as the 'adiabatic timescale'. The total evolution time is adiabatic if it satisfies $t_f \gg \theta$. Thus, the system evolves adiabatically (diabatically) if the diabatic evolution bound is small (large).

This version of the adiabatic theorem amounts to finding an expression for $\Uad$, that contains information about the dynamic and geometric phase acquired along the evolution, and can be found in the book \cite{Teufel:book} for unbounded operators. Note that typical textbook expressions (e.g., Ref.~\cite{Messiah}) just bound the overlap between $\Uad (1)|\psi(0)\rangle$ and the final state $\Utot(1)\ket{\psi(0)}$, where $\ket{\psi(0)}$ is the lowest eigenstate of $H(0)$. Instead, we consider any initial state $|\psi(0)\rangle \in P(0) \mathcal{H}$, not just the ground state, and also compute the total phase. This is also more flexible in that, in fact, the projector $P$ can single out any subspace of eigenstates of $H$ (not necessarily the lowest), which may or may not be degenerate.

Techniques exist to improve the bound to $\gamma_k/t^k_f$ for integers $k>1$. This is done by requiring the time-dependent Hamiltonian to have vanishing derivatives up to order $k$ at the initial and final time~\cite{Garrido:62}, or just the final time in the case of an open system~\cite{Venuti:2018aa}. It is even possible to make the bound exponentially small in $t_f$~\cite{Nenciu:93,Hagedorn:02,lidar:102106,RPL:10,Cheung:2011aa,Ge:2015wo}. We will not be concerned with this problem here; instead, we focus on providing an \textit{explicit} expression for the constant $\theta$. We are particularly interested in presenting an expression for $\theta$ that is finite even when used beyond the scope of our proof for some unbounded Hamiltonian $H_\infty(s)$. A paradigmatic example of such a system is a (perturbed) harmonic oscillator whose Hamiltonian changes slowly with time. Such systems are common in quantum computation, e.g., in the context of effective Hamiltonians of superconducting circuits~\cite{Yurke:1984aa} and in describing the coupling between trapped ions via motional degrees of freedom~\cite{Cirac:95}.

A large body of work exists on proofs of Eq.~\eqref{eq:diff}, including for unbounded Hamiltonians, starting with the work of Kato~\cite{Kato:50}, who improved upon the original proof of Born and Fock for simple discrete spectra~\cite{born_beweis_1928} (e.g. a one-dimensional harmonic oscillator), allowing $P(s)$ to be a finite-dimensional spectral projection associated with an isolated eigenvalue (e.g. the hydrogen atom).
Subsequent works, e.g., Ref.~\cite[Sec.~5]{Avron:87} and Refs.~\cite{klein_power-law_1990,Avron_1998,avron_adiabatic_1999,Teufel:book,Jansen:07,Venuti:2015kq}, explored many possible generalizations and refinements of this result, but to the best of our knowledge  a recipe for actually computing the number $\theta$ for a specific unbounded system has not yet been provided.
In order to keep our results accessible to physicists seeking to estimate $\theta$, we use a somewhat non-traditional approach to unbounded Hamiltonians such as the harmonic oscillator or the hydrogen atom. The traditional approach uses abstract mathematical concepts to rigorously and directly work with an allowed family of unbounded Hamiltonians, and is extensively discussed in the literature (see in particular Ref.~\cite{Schmid:2014} for the most general family), but notably lacks calculations for concrete examples or accessible estimates that can be used in specific cases. We note that often the Hamiltonians used in physics can be restricted to finite-dimensional Hilbert spaces after introducing appropriate cutoffs. All our proofs and results concern finite-dimensional bounded Hamiltonians obtained after such cutoffs. Such truncations are common in numerical simulations of experimental systems of the type that our results are designed to be applied to, e.g. Ref.~\cite{khezri2021customized}.

The specific way the cutoff $\Lambda$ is introduced depends on how the initially unbounded Hamiltonian is given to us. If it is provided along with a countable set of basis vectors $\{|n\rangle\},~ n=0,1,2\dots$, and the matrix elements $H_{nm}$ are given explicitly as functions of $n$ and $m$, then just restricting the matrix $H$ to $0\leq n \leq \Lambda -1$ 
provides a Hamiltonian with a cutoff $\Lambda$ that will feature in our results.
If the initially unbounded Hamiltonian is instead provided via operators corresponding to continuous variables, such as momentum and position for a particle on the line, then one must choose an appropriate countable basis, e.g., the harmonic oscillator basis for each of the dimensions. After that, it is straightforward to form a countable basis set and compute matrix elements $H_{nm}$, making the 
cutoff $\Lambda$ applicable as in the first case. Alternatively, one may discretize one of the conjugate variables for each dimension, obtaining a $\Lambda\times\Lambda$ matrix $H$ where each matrix element in principle depends on the step of the discretization grid and thus on $\Lambda$.

We seek an expression for $\theta$ that does not diverge with $\Lambda$ even when the finite-dimensional operator norm $\|H_\Lambda'\|$ may diverge with the cutoff. The adiabatic timescales for unbounded Hamiltonians available in the literature (e.g., see Ref.~\cite[Eq.~(2.2)]{Teufel:book}) achieve this by using a different norm for $H'$, that is free of the cutoff divergence. This expression for $\theta$ is not readily usable for analytic estimates, and it requires much work and prior knowledge for evaluation in a numerical simulation. Here, we resolve this issue by obtaining an analytically tractable expression for $\theta$, interpret the norm of $H'$ used in \cite{Teufel:book} in physical terms, and improve upon it by replacing it by $\|H'P\|$ almost everywhere.%
\footnote{Note that our definition of a cutoff is basis-dependent. It is also important to choose the subspace of interest $P$ consistently for each $\Lambda$. Consider the case of a time-dependent harmonic oscillator $p^2 + \omega^2(t)x^2$, with eigenstates of $p^2 + \omega^2(0)x^2$ used as the basis. While our bound will still technically hold for any choice of $P$, choosing $P$ to project on the highest energy state after the cutoff will lead to a diverging adiabatic timescale. Indeed, $\|PH'Q\|$, where $Q=I-P$, will grow with $\Lambda$. The general methodology of the  choice of basis and $P$ is outside of the scope of this work.}
Moreover, our $\theta$ remains small even for an exponentially large dimension $d$ of the subspace $P\mathcal{H}$, and we make the dependence on the gap $\D$ explicit.

This same approach will also allow us to address the problem of deriving an error bound on the evolution generated by \textit{effective} Hamiltonians $H_{\text{eff}}$ that are operators in a smaller Hilbert space corresponding to the low-energy subspace of the original problem. We identify the isometry $V(s)$ into that Hilbert space and the matrix $H_{\text{eff}}$, such that the solution of the Schr\"odinger equation
$u'(s) = -it_f H_{\text{eff}}(s) u(s)$ with $u(0) = I$ is close to the true evolution due to the same adiabatic theorem stated above:
\begin{equation}
    \|u(s) - V(s)U_{\text{tot}}(s)V^{\dag}(s) \| \leq b \ .
    \label{eq:4}
\end{equation}
We apply our results to circuits of superconducting flux qubits \cite{Mooij:99,Wendin:2017aa}, of the type used, e.g., in quantum annealing~\cite{Harris:2008lp,Yan15a,khezri2020annealpath}. Quantum annealing (reviewed in Refs.~\cite{RevModPhys.80.1061,Albash-Lidar:RMP,Hauke:2020,Chakrabarti:2022}) is a field primarily studying heuristic quantum algorithms for optimization, best suited to run on analog quantum devices. In the qubit language, the quantum annealer is typically initialized in a uniform superposition state that is the ground state of a transverse field Hamiltonian. Over the course of the algorithm, the strength of the transverse field is gradually decreased while simultaneously the strength of the interactions encoding the optimization problem of interest is gradually increased, guiding the quantum evolution towards the ground state that encodes an optimal solution. In the context of superconducting devices, the qubits used for this, with frequency $\omega_{\text{q}}$, are described by a circuit model (which includes capacitors, Josephson junctions, etc.), characterized by the capacitive energy $E_C$ and the Josephson junction energy $E_J \gg E_C$. We express the plasma frequency $\omega_{\text{pl}}(s)$ and the residual transverse field $\omega_{\text{q}}\delta$ at the end of the anneal via the circuit parameters $E_J,E_C$ and the schedule of the control fluxes. We obtain a bound for the adiabatic timescale $\theta$ in Eq.~\eqref{eq:diff}: $\omega_{\text{q}} \theta = O(\omega_{\text{q}}/(\omega_{\text{pl}}(1) \delta))(\ln\frac{\omega_{\text{pl}}(1) }{ \omega_{\text{q}}\delta})^{-1}$, while applying the existing analytically tractable form of the adiabatic theorem~\cite{Jansen:07} yields $\omega_{\text{q}} \theta =\Theta(\Lambda)$,\footnote{The big-$\Theta$ notation is defined by $y =\Theta(x) ~ \Leftrightarrow ~ (y=O(x)$ and $x=O(y)$) which includes proportionality up to a constant.} which diverges with the cutoff. We also check that for finite $\|H'\|$ the existing form~\cite{Jansen:07} gives a result that is consistent with our bound, namely: $\omega_{\text{q}} \theta = O(\omega_{\text{q}}/(\omega_{\text{pl}}(1) \delta))$.
For these expressions written in terms of $E_J$ and $E_C$ see Sec.~\ref{EjEcSec}. Thus, our results include the first non-diverging expression for the adiabatic timescale in the case of unbounded Hamiltonians, as well as a new practical application of existing rigorous forms of the adiabatic theorem.

The structure of the rest of this paper is as follows. We provide detailed definitions required to state our result, as well as compare it with previous work, in Sec.~\ref{sec:result}. The paper is written in a way that allows the reader to skip the proof that follows this section, and move on to applications in Sec.~\ref{sec:fluxQ}. The proof is given in two parts: a short argument for obtaining an $O(1/t_f)$ bound in Sec.~\ref{sec:bigO} and a lengthier Sec.~\ref{ResSec} in which we compute the constant $\theta$. The application to flux qubits can be found in Sec.~\ref{sec:fluxQ}, which is also separated into results and a proof that can be skipped. We give the definition of the effective (qubit) Hamiltonian in Sec.~\ref{eHam}, along with a discussion of how the adiabatic theorem bounds we obtained apply in the effective Hamiltonian setting. Sec.~\ref{sec:fluxQ} and Sec.~\ref{eHam} are independent of each other. We conclude in Sec.~\ref{sec:conc}. Additional calculations in support of the flux qubit analysis are presented in Appendix~\ref{app:A}, and a proof of the intertwining relation is given in Appendix~\ref{app:intertwining}.

\section{Adiabatic and diabatic evolution}
\label{sec:result}

\subsection{Previous work}
\label{graphNormHere}
To set the stage for our results on the adiabatic theorem, we first briefly review key earlier results. We note that unlike these earlier works, we will provide an explicit expression for the adiabatic timescale, that does not diverge with the cutoff of the Hamiltonian in most relevant examples and is ready to be used both analytically and numerically. This is an important aspect of the novelty of our contribution to the topic. 

Such a ready-to-use result was obtained for finite-dimensional (bounded) Hamiltonians by Jansen, Ruskai, and Seiler (JRS), and our results closely follow their work.  They prove several bounds, including~\cite[Theorem 3]{Jansen:07}:
\begin{quote}
Suppose that the spectrum of $H(s)$ restricted to $P(s)$ consists of $d(s)$ eigenvalues (each possibly degenerate, crossing permitted) separated by a gap $2\D(s)$ from the rest of the spectrum of $H(s)$, and $H$, $H'$, and $H''$ are bounded operators. Let $P_{t_f}(s)\equiv \Utot(s) P(0) \Utot^\dag(s)$. Then 
\bes
\label{eq:JRS-AT}
\begin{align}
\|P_{t_f}(s^*)-P(s^*)\| &<\frac{\theta(s^*)}{t_f} \\ 
\theta(s^*) &= \left.\frac{d\|H'\|}{\D^2}\right|_{s=0} +\left.\frac{d\|H'\|}{\D^2}\right|_{s=s^*} + \int_0^{s^*}\left(\frac{d\|H''\|}{\D^2} + 7d\sqrt{d}\frac{\|H'\|^2}{\D^3} \right)ds ,
\end{align}
\ees
\end{quote}
The direct dependence on $\|H'\|$ and $\|H''\|$ is the crucial one from our perspective, and the one we avoid in this work. Indeed these norms diverge with the cutoff for, e.g., a time-dependent harmonic oscillator or the hydrogen atom.

The adiabatic timescale that is harder to use analytically and numerically can be found in Ref.~\cite[Eq.~(2.2)]{Teufel:book}:
\begin{equation}
    \theta  = 
    \|F(0)\| +\|F(1)\| +\int_0^{1} \|F[P',P]\| + \|F'\|ds  ,
\end{equation}
where
\begin{equation}
    F  = \frac{1}{2\pi i} \oint_\Gamma QR(z)  R'(z)dz  + h.c.\ , \quad R(z)= (H-z)^{-1}
\end{equation}
and $\Gamma$ is a contour around the part of the spectrum corresponding to $P\mathcal{H}$. In what follows we give a simplifed non-rigorous summary of the arguments used in Ref.~\cite{Teufel:book} to prove that $\theta < \infty$. The boundedness of the norm of $F$ and its derivative can be traced down to an assumption:
\begin{equation}
    \forall |\phi\rangle,~ \|\phi\|=1:~ \|H'|\phi\rangle\|^2 \leq C_\varepsilon^2 (1  + \|H|\phi\rangle\|^2/\varepsilon^2) \ , \label{c1b}
\end{equation}
where we kept an energy scale $\varepsilon$ to match the dimensions, but $\varepsilon =1$ is usually taken in the mathematical literature. The smallest such constant $C_\varepsilon =\|H'\|_{\mathcal{L}(\mathcal{D},\mathcal{H})}$ is actually the definition of the operator norm for unbounded Hamiltonians with a domain $\mathcal{D}$. The space $\mathcal{D}$ is equipped, besides the usual state norm $\|\psi \|_\mathcal{H}$ inherited from $\mathcal{H}$, with a different state norm $\|\psi \|_\mathcal{D}$ than $\mathcal{H}$, called the graph norm:
\begin{equation}
    \|\psi \|_\mathcal{D} = \sqrt{ \|\psi \|_\mathcal{H}^2 +  \|H_0 \psi \|_\mathcal{H}^2 /\varepsilon^2} 
\end{equation}
for some Hamiltonian $H_0$ (that we take $=H$ for a tighter bound) and some arbitrary energy scale $\varepsilon$. The operator norms are now computed with respect to the spaces they map between:
\begin{equation}
    ||O||_{{\mathcal{L}(\mathcal{A,B})}}  =\text{sup}_{\psi \in \mathcal{A},~ \psi\ne 0 } \frac{||O\psi||_\mathcal{B}}{||\psi||_\mathcal{A}}
\end{equation}
Using this definition $\|H'\|_{\mathcal{L}(\mathcal{D},\mathcal{H})}$ is now a finite number $C_\varepsilon$ under the assumption~\eqref{c1b}. This assumption is commonly used to prove that a solution of the time-dependent Schr\"odinger equation exists, with the difference that a single Hamiltonian $H_0$ is used on the r.h.s. for all times. More importantly, since the resolvent is formally defined by $(H-z)R(z) = 1_\mathcal{H}$ as a map from $\mathcal{H}$ to $\mathcal{D}$, one can prove that the usual  operator norm of the combination $H'R(z)$ is bounded as:
\begin{equation}
   \| H'R(z) \| \le \| H'\|_{\mathcal{L}(\mathcal{D},\mathcal{H})} \|R(z) \|_{\mathcal{L}(\mathcal{H},\mathcal{D})} .
\end{equation}
At the cost of the small increase in the norm of the resolvent, we got a finite number $C_\varepsilon$ in place of the norm of the unbounded operator. Using this idea, \cite{Teufel:book} proves $\theta < \infty$.
Note that for finite-dimensional systems the assumption~\eqref{c1b} can also be written as:\footnote{Our matrix inequalities have the standard meaning: $A\leq B ~ \Leftrightarrow ~ B-A$ has nonnegative eigenvalues.} 
\begin{equation}
    H'^2 \leq C_\varepsilon^2  + C_\varepsilon^2 H^2/\varepsilon^2= c_0 +c_1 H^2 .
\end{equation}
The adiabatic timescale $\theta$ of \cite{Teufel:book} depends on $c_0$ and $c_1$, or equivalently on $ \| H'R(z) \|$, and the same quantities for the second derivative $H''$ coming from $\|F'\|$, though the dependence is never explicitly computed. Here, we will be able to remove the dependence on the constants coming from $H''$, and replace most of the appearances of $c_0$ and $c_1$ by a tighter bound. The physical meaning of the norm $\|O\|_{\mathcal{L}(\mathcal{D},\mathcal{H})}$ is as follows: given a state $|\psi\rangle$ with a bounded expectation value of energy $ \langle \psi| H|\psi \rangle\leq E$, the norm $\|O\|_{\mathcal{L}(\mathcal{D},\mathcal{H})}^2$ is the maximal value of $ \langle \psi| O^2|\psi \rangle/ (1 + E^2/\varepsilon^2)$.

The earlier work by Avron and Elgart~\cite[Sec.~5]{avron_adiabatic_1999}, while mainly focusing on gapless bounded Hamiltonians, discussed the adiabatic theorem for unbounded gapless Hamiltonians. They required that both the resolvent $R(z=i,s)$ and $H(s)R'(z=i,s)$ are bounded. Essentially the same assumption was made by Abou Salem~\cite[Sec.~2]{salem_quasi-static_2007}, in the context of non-normal generators. 

Recent work~\cite{Bachmann:2017aa,Fraas:18} presents a refinement of the adiabatic theorem for a different case of diverging $\|H'\|$ that comes from the thermodynamic limit of the size of a many-body spin system. While the authors do not present an explicit form for $\theta$, we believe their methods are an alternative way to remove the dimension $d$ of the subspace $P\mathcal{H}$, and in fact any dependence on the system size from the bound on local observables.

\subsection{Adiabatic intertwiner}

Following Kato~\cite{Kato:50}, we define an approximate evolution in the full Hilbert space $\mc{H}$:
\begin{equation}
    \Uad '(s) = -i\Had(s) \Uad (s), \quad \Uad (0) = I, \quad s\in[0,1] 
    \label{eq:intertwiner}
\end{equation}
where $\Uad $ is called the \textit{adiabatic intertwiner} and the (dimensionless) adiabatic Hamiltonian is
\beq
\Had(s) = t_f H(s) + i[P'(s),P(s)] . 
\label{eq:Had}
\eeq
Note that both $\Had$ and $\Uad$ are $t_f$-dependent. Here $P(s)$ is a finite-rank projection on the low-energy subspace of $H(s)$ (i.e., the continuous-in-$s$ subspace spanned by the eigenvectors with the lowest $d(s)$ eigenvalues\footnote{The number $d_P$ of these eigenvectors is thus constant and equal to the the dimension of the subspace. Allowing for degeneracy, $d_P \geq d(s)$, and we use $d(s)$ until Sec.~\ref{eHam}, at which point we switch to $d_P$.}). A property of this approximation is that the low-energy subspace is preserved:
\begin{equation}
    \Uad (s)P_0= P(s) \Uad (s)
    \label{eq:intertwining}
\end{equation}
where here and henceforth we denote $P(0)$ by $P_0$, and drop the $s$ time-argument from $P(s)$ where possible.
The proof of this intertwining property is well-known and has been given many times in various forms and subject to various generalizations;
see, e.g., Refs.~\cite{Avron:87,salem_quasi-static_2007,PhysRevA.77.042319,ABLZ:12-SI,Venuti:2015kq}, as well as our Appendix~\ref{app:intertwining}. The idea (due to Kato~\cite{Kato:50}, who presented the original proof; see his Eq.~(22)) is to show that both sides solve the same
initial value problem, i.e., equality holds at $s=0$, and they satisfy the same differential equation after differentiating by $s$. The latter can be shown using Eqs.~\eqref{eq:P'-offD0} and~\eqref{eq:P'-offD} below. 

The operator $P'$ has the following useful properties. Since $P^2 = P$, we have 
\beq
P' =P'P + PP' .
\label{eq:P'-offD0}
\eeq 
Multiplying by $P$ on the right, and letting $Q\equiv I-P$, we obtain $QP'P = P'P$, i.e., 
\beq
PP'P=0 \ , \quad QP'Q=0 ,
\label{eq:P'-offD}
\eeq 
where the proof of $QP'Q =0$ is similar. Thus $P'$ is block off-diagonal:
\beq
P' = PP'Q +QP'P .
\label{eq:P'offdiag}
\eeq

We also note that for a spatially local system the generator related to $i[P',P]$ is approximately a sum of local terms \cite{osborne2007simulating}. This approximation is known as a {\textit{quasiadiabatic continuation}} \cite{HastingsQuasiAdiabatic}, though we will not discuss locality in this work.  

\subsection{Bounds on states and physical observables}
\label{phys_o}

We would like to bound certain physical observables via the quantity $b$ defined in Eq.~\eqref{eq:diff}. Since $b$ bounds the difference between the actual and adiabatic evolution, we refer to $b$ as the `diabatic evolution bound'.

We note that Kato's adiabatic theorem~\cite{Kato:50} established, for bounded Hamiltonians, that the quantity $[\Uad (s) -  \Utot(s)]P_0$ tends to zero as $1/t_f$, but it will still take us most of the rest of this paper to arrive at the point where we can state with conviction that the bound in Eq.~\eqref{eq:diff} does not diverge with the cutoff. This will require extra assumptions; indeed there are contrived unbounded Hamiltonians where Kato's quantity is arbitrarily large for any finite evolution time $t_f$.

Note that using unitary invariance we can rewrite Eq.~\eqref{eq:diff} as $\|x(s)\|\leq b$, where
\beq
x(s) \equiv P_0\Uad^\dag (s)\Utot(s) - P_0 , 
\label{eq:x}
\eeq

\subsubsection{Bound on the final state difference}

Consider an initial state $|\phi\rangle$ in the low-energy subspace ($P_0|\phi\rangle =|\phi\rangle$). We wish to compare the evolution generated by $\Utot$ with the one generated by $\Uad $. Dropping the $s$ time-argument from the $U$'s, the difference in the resulting final states is:
\bes
\label{eq:13}
\begin{align}
\label{eq:13a}
    \|(\Uad  - \Utot)|\phi\rangle\|^{2} &= \|(\Uad  - \Utot)P_0|\phi\rangle\|^{2} =\langle \phi| ((\Uad  - \Utot)P_0)^\dag(\Uad  - \Utot)P_0 |\phi\rangle \\
\label{eq:13b}
    & \leq \|((\Uad  - \Utot)P_0)^\dag(\Uad  - \Utot)P_0\|  \leq \|(\Uad  - \Utot)P_0\|^{2} = \| (\Utot^\dag \Uad  - I)P_0\|^2 \\
\label{eq:13c}
    &= \| ((\Utot^\dag \Uad  - I)P_0)^\dag\|^2 =\| P_0(\Uad^\dag\Utot   - I)\|^2 =     \| x\|^2 \leq b^2 .
\end{align}
\ees
We use this quantity because we would like to describe the error in both the amplitude and the acquired phase of the wavefunction.

\subsubsection{Bound on leakage}

If we are just interested in the leakage from the low-lying subspace, it can be expressed as:
\begin{equation}
    P_{\text{leak}} = \langle \phi |\Utot^\dag Q_0\Utot |\phi\rangle  = \langle \phi |(Q_0\Utot P_0)^\dag Q_0\Utot P_0 |\phi\rangle\leq \|Q_0\Utot P_0\|^2 .
\end{equation}
Then:
\begin{align}
  \|Q_0\Utot P_0\| =  \|Q_0(\Utot - \Uad) P_0\| \leq  \|(\Utot - \Uad) P_0\| = \|( \Uad^\dag\Utot -I) P_0\| = \|x\|\leq b ,
 \end{align} 
so that
\begin{align}
  P_{\text{leak}}  \leq b^2 .
\end{align}

\subsubsection{Bound on the error in an observable $O$}
\label{Obound}
The expectation value for an observable $O$ in the evolved state $\Utot|\phi\rangle $ as opposed to the approximate state $\Uad |\phi\rangle $ is different by:
\begin{equation}
   \langle \phi| \Utot^\dag  O\Utot|\phi\rangle -\langle \phi| \Uad ^{\dag}  O\Uad |\phi\rangle  \leq 2b\|O\| .
   \label{eq:Obound}
\end{equation}
To prove this, note that:
\begin{equation}
   \Utot|\phi\rangle = \Uad |\phi\rangle   + \Delta_U |\phi\rangle, \quad \Delta_U \equiv \Utot -\Uad  , \quad  \Delta_U |\phi\rangle =  -\Utot x^\dag|\phi\rangle  .
\end{equation}
Therefore:
\bes
\begin{align}
\langle \phi| \Utot^\dag  O\Utot|\phi\rangle -\langle \phi| \Uad ^{\dag}  O\Uad |\phi\rangle &= \bra{\phi}\Uad^\dag O \Delta_U\ket{\phi} + \bra{\phi}\Delta_U^\dag O \Utot\ket{\phi}  \\
& \leq \| \Uad^\dag O \Utot x^\dag \| + \| x \Utot^\dag O \Utot \|  \\
& \leq \|O\| \left(\|x^\dag\| + \|x\|  \right)
\end{align}
\ees
from which Eq.~\eqref{eq:Obound} follows.

One of the immediate consequences is that measuring $Z$ (or any other unit-norm observable) on one qubit in a $n$ qubit system after the evolution can be described by an approximate evolution $\Uad $ to within an error of $(2b+b^2)$ in the expectation value.

\subsubsection{Bound on the JRS quantity}
The quantity appearing in the JRS bound~\eqref{eq:JRS-AT} satisfies
\beq
\|P_{t_f}-P\|  = \|  \Utot P_0\Utot^\dag  - \Uad  P_0\Uad^{\dag} \| = \|  \Uad^\dag\Utot  P_0 -   P_0\Uad^{\dag} \Utot\| = \| Q_0 \Uad^\dag\Utot  P_0 -   P_0\Uad^{\dag} \Utot Q_0\|  ,
\label{eq:20}
\eeq 
where in the last equality we used $Q_0=I-P_0$ and added/subtracted $P_0\Uad^{\dag} \Utot P_0$.

Using the definition of $x$ [Eq.~\eqref{eq:x}], we can express:
\beq
  P_0\Uad^\dag \Utot =  P_0  + x, \quad \Uad^\dag\Utot  P_0 =P_0 - \Uad^\dag\Utot x^\dag ,
\eeq
so that Eq.~\eqref{eq:20} becomes:
\beq
\|P_{t_f}-P\|  =  \| Q_0 \Uad^\dag\Utot x^\dag  P_0 +   P_0x Q_0\| = \max(\|\Uad^\dag\Utot x^\dag\|, \|x\|) = \|x\| \leq b  ,
\label{eq:connect-to-JRS}
\eeq 
where the second equality follows since $Q_0 \Uad^\dag\Utot x^\dag  P_0$ and $P_0x Q_0$ are two opposite off-diagonal blocks and their eigenvalues do not mix, and the last equality follows from the unitary invariance of the operator norm.

We proceed to explicitly express the bound $b$ in the next section.

\subsection{Statement of the theorem}
Collecting the definitions of the previous sections, we present our main result:

\begin{mytheorem}[Adiabatic theorem]
\label{th:AT}
Assume that $\forall s\in[0,1]$ there exist positive numbers $c_0, c_1$ such that the Hamiltonian $H(s)$ satisfies 
\begin{equation}
    H'^2 \leq c_0 + c_1 H^{2}\ .
    \label{eq:main-assump}
\end{equation}
Let $P(s)$ denote the projection onto a continuous-in-$s$  eigensubspace of the Hamiltonian $H(s)$ corresponding to $d(s)$ eigenvalues, that occupies an interval $r(s)$ in energy centered around zero energy and is separated by a gap $2\D(s)$ from all other eigenvalues; see Fig.~\ref{fig:contour}. Assume that the initial state $\ket{\phi} \in P(0)\equiv P_0$.  Then the adiabatic intertwiner $\Uad$ [the solution of Eq.~\eqref{eq:intertwiner}] satisfies the following bounds on its difference with the true evolution $\Utot$:
\begin{equation}
   \|P_0\Uad^\dag \Utot - P_0\| \leq b, \quad \|(\Uad  -  \Utot)P_0\|  \leq b ,\quad \|(\Uad  - \Utot)|\phi\rangle\| \leq b , \quad \|  \Utot P_0\Utot^\dag  - \Uad  P_0\Uad^{\dag} \| \leq b ,
   \label{eq:ineq-th}
\end{equation}
where $b = \theta/t_f$ and $\theta$ is given by

\begin{align}
    & \quad \theta =  \tau^2(0)\|P_0H'(0)Q_0\| +\tau^2(s^*)\|P(s^*)H'(s^*)Q(s^*)\| + \int_0^{s^*}ds [\tau^3(5\|PH'Q\| + 3 \|PH'P\|)\|PH'Q\| \notag \\
    &\qquad \qquad + \tau^2\|PH''Q\|  + 3\tau^3\sqrt{\sum_{k=0}^{1} c_k\|PH'H^kQ\|^2}]   ,
\label{adtime}
\end{align}

Here $t_f$ is the total evolution time, $s^*\in[0,1]$ is the final value of $s$ and
\beq
\tau = \min\left(\frac{\sqrt{d(s)}}{\D(s)},\frac{2r(s)+ 2\pi \D(s)}{2\pi \D^2(s)}\right) .
\label{eq:tau}
\eeq
Another valid $\theta$ can be obtained from Eq.~\eqref{adtime} by replacement:
\begin{equation}
    \tau\|PH'HQ\| ~ \to ~ \|PH' Q \|\left( 1 + \min\left(\frac{\sqrt{d(s)}}{2\D(s)}r(s),\frac{2r(s)+ 2\pi \D(s)}{4\pi \D^2(s)}(r(s)+2\D(s))\right)\right) .
\end{equation}
\end{mytheorem}
Note that the first three inequalities stated in Eq.~\eqref{eq:ineq-th} were already established in Eqs.~\eqref{eq:13}, and the last in Eq.~\eqref{eq:20} along with Eq.~\eqref{eq:connect-to-JRS}. The new aspect of Theorem~\ref{th:AT} is the value of the bound $\theta$, which does not involve $\|H'\|$ or higher derivatives that may diverge with the cutoff used to define $H(s)$. Moreover, 
$\|PH'Q\|$ gives a tighter bound than $\|H'\|_{\mathcal{L}(\mathcal{D},\mathcal{H})}$ that would have appeared from the direct translation of the adiabatic theorem for unbounded Hamiltonians given in Ref.~\cite{Teufel:book}. Indeed, 
\bes
\begin{align}
 \|PH'Q\| &=\|QH'P\| \leq \|H'P\| =\text{max}_{\psi \in P\mathcal{H}, \|\psi\| =1} \|H'\psi\| \leq \|H'\|_{\mathcal{L}(\mathcal{D},\mathcal{H})} \sqrt{ 1 + \|H \psi\|^2/\varepsilon^2} \\
 & \leq  \|H'\|_{\mathcal{L}(\mathcal{D},\mathcal{H})} \sqrt{ 1 + r(s)^2/4\varepsilon^2} .
 \end{align}
\ees
In terms of $c_0,c_1$: $\|PH'Q\| \leq \sqrt{c_0 + c_1 r(s)^2/4}$. When the above inequalities are tight, our bound would match the one that could in principle be obtained from \cite{Teufel:book}. However in many relevant cases such as a harmonic oscillator with small time-dependent anharmonicity $\|PH'Q\|$ is parametrically less than the r.h.s. We also find the form of $PH'Q$ to be more insightful than $\|H'\|_{\mathcal{L}(\mathcal{D},\mathcal{H})}$.

Since the constants $c_0$ and $c_1$ depend on the choice of the constant energy offset, we chose zero energy to lie in the middle of the eigenvalues corresponding to $P\mathcal{H}$. We note that for bounded $H'$ the assumption~\eqref{eq:main-assump} is automatically satisfied with $c_1(s) =0$ and $c_0(s) = \|H'\|^2$, since ${H'}^2-\|H'\|^2 I \leq 0$ (a negative operator) by definition of the operator norm. Using this, we can reduce Eq.~\eqref{adtime} to a form that depends on $\|H'\|$, which allows direct comparison to Eq.~\eqref{eq:JRS-AT} (from Ref.~\cite{Jansen:07}) using $\tau = \frac{\sqrt{d}}{\D}$:
\begin{mycorollary} 
\label{cor:JRScomp}
The JRS adiabatic timescale $\theta^{\text{JRS}}(s^*)$ and the weaker version of our new adiabatic timescale $\theta^{\text{new}}(s^*)$ are:
\bes
\label{eq:thetas}
\begin{align}
\label{eq:theta-JRS}
\theta^{\text{JRS}}(s^*) &= \left.\frac{d\|H'\|}{\D^2}\right|_{s=0} +\left.\frac{d\|H'\|}{\D^2}\right|_{s=s^*} + \int_0^{s^*}\left(\frac{d\|H''\|}{\D^2} + 7d\sqrt{d}\frac{\|H'\|^2}{\D^3} \right)ds ,
\\
\theta^{\text{new}}(s^*) &= \left.\frac{d\|PH'Q\|}{\D^2}\right|_{s=0} +\left.\frac{d\|PH'Q\|}{\D^2}\right|_{s=s^*} +\notag \\
&+ \int_0^{s^*}\left(\frac{d\|PH''Q\|}{\D^2} + d\sqrt{d}\frac{\|PH'Q\|(5\|PH'Q\| + 3\|PH'P\| + 3\|H'\|)}{\D^3} \right)ds ,
\label{eq:theta-new}
\end{align}
\ees
\end{mycorollary}
We see that though our new adiabatic timescale has slightly larger numerical coefficients, the projected form of the operators can provide a qualitative improvement over the JRS result.\footnote{We emphasize that $\theta^{\text{new}}$ did not appear in the derivation of $\theta^{\text{JRS}}$, though some intermediate formulas arrived at in Ref.~\cite{Jansen:07} may seem similar at first glance. The derivation of $\theta^{\text{JRS}}$ involves bounds on $\|P'\|$, whereas in our case $P$ and $Q$ do not involve derivatives and serve to reduce the norm of $H'$ or $H''$ in between them.} Note that we can also write a bound that is free of the dimension $d$ if the second option for $\tau$ in Eq.~\eqref{eq:tau} is smaller than the first.

\section{Diabatic evolution bound}
\label{sec:bigO}

We will calculate a diabatic evolution bound $b$ on the quantity in Eq.~\eqref{eq:diff} for some  $s^* \in [0,1]$:
\begin{equation}
    \|[\Uad(s^*) -  \Utot(s^*)]P(0)\| = \| f(s^*) - P_0\| ,
   \label{eq:18}
\end{equation}
where 
\begin{equation}
   f(s) \equiv P_0\Uad^\dag (s)\Utot(s) = x(s)+P_0 .
\end{equation}
We would like to express $f(s^*)$ via  $f'(s)$:
\begin{equation}
    f(s^*) = P_0 + \int_0^{s^*} f'(s)ds 
\label{eq:fs*}
\end{equation}
Recalling that $\Utot$ satisfies Eq.~\eqref{eq:exact} and
$\Uad $ satisfies Eq.~\eqref{eq:intertwiner}, the derivative is:
\begin{align}
     f'(s) =P_0 ( { \Uad^\dag }'\Utot + \Uad^\dag \Utot')
     = P_0 \Uad^\dag (it_f H  -[{P'}^\dag,P] - it_f H )\Utot 
     = -P_0 \Uad^\dag [{P}',P]\Utot 
     \label{stopPoint}
\end{align}
where we used ${P'}^\dag = P'$.
Note how the $O(t_f) $ term cancelled, so the expression appears to be $O(1)$. However, it is in fact $O(1/t_f) $, as we show next. 

For any operator $X(s)$ define $\tilde{X}(s)$ (``twiddle-$X$")~\cite{Avron:87} such that 
\begin{equation}
    [X(s),P(s)] = [H(s),\tilde{X}(s)] ,
    \label{eq:tilde}
\end{equation}
and the diagonal of $\tilde{X}$ in the eigenbasis of $H(s)$ is zero. Note that $\tilde{X}$ has units of time relative to $X$.

For instance, ${P'}^\sim$ is defined via:\footnote{Our convention is that the tilde takes precedence over derivatives, i.e., $\tilde{X}'\equiv (\tilde{X})'$. When the derivative is to be taken first we write the tilde to the right of the operator, i.e., ${X'}^\sim\equiv (X')^\sim$.}
\begin{equation}
        [P'(s),P(s)] = [H(s),{P'}^\sim(s)]
\end{equation}
The details of why $\tilde{X}$ exists and how it is expressed via $X$ are given in Sec.~\ref{ResSec}. 
Proceeding with bounding Eq.~\eqref{stopPoint}, we can now rewrite it as:
\begin{align}
 f'(s) =  -P_0 \Uad^\dag [H,{P'}^\sim]\Utot .
\end{align}
Note that, using Eqs.~\eqref{eq:exact} and~\eqref{eq:intertwiner}:
\begin{equation}
    ( \Uad^\dag {P'}^\sim\Utot)' =  \Uad^\dag  (it_f H {P'}^\sim - [P',P]{P'}^\sim  + {{P'}^\sim}'-{P'}^\sim it_f H ) \Utot ,
\end{equation}
which we can rearrange as:
\begin{equation}
     \Uad^\dag [H,{P'}^\sim ]\Utot  = \frac{1}{it_f} [ ( \Uad^\dag {P'}^\sim \Utot)' +  \Uad^\dag  ([P',P]{P'}^\sim   - {{P'}^\sim }')\Utot ] .
\end{equation}
Using this in Eq.~\eqref{stopPoint}, we obtain the desired $O(1/t_f) $ scaling:
\begin{align}
     f'(s) = \frac{iP_0}{t_f} [ ( \Uad^\dag {P'}^\sim \Utot)' -  \Uad^\dag  (P'{P'}^\sim   + {{P'}^\sim }')\Utot] ,
\end{align}
where using Eq.~\eqref{eq:intertwining} we simplified one term in the commutator as $P_0 \Uad^\dag P= P_0 \Uad^\dag$, and also using Eq.~\eqref{eq:P'-offD}, we have $P_0 \Uad^\dag P' P= \Uad^\dag P P' P =0$, so that the other term with $P'P$ in the commutator vanishes. Plugging this back into Eq.~\eqref{eq:fs*}, we get:
\begin{align}
    f(s^*) -P_0 =\frac{iP_0}{t_f}   \left( ( \Uad^\dag {P'}^\sim \Utot)|_0^{s^*} -\int_0^{s^*}  \Uad^\dag  (P'{P'}^\sim   + {{P'}^\sim }')\Utot ds \right) .
\end{align}

Using $P_0 \Uad^\dag = \Uad^\dag P$ throughout, this results in the following bound on the quantity in Eq.~\eqref{eq:18} we set out to bound:
\bes
    \label{eq:deb1}
\begin{align}
    \|[\Uad(s^*) -  \Utot(s^*)]P(0)\| &= \| f(s^*)-P_0 \| \leq b = \frac{\theta}{t_f} \\
    \theta & = 
    \|P_0{P'}^\sim (0)\| +\|P(s^*){P'}^\sim (s^*)\| +\int_0^{s^*} \|PP'{P'}^\sim \| + \|P{{P'}^\sim }'\|ds  .
\end{align}
\ees
The adiabatic timescale $\theta$ given here is not particularly useful in its present form. Thus, we next set out to find bounds on each of the quantities involved. Our goal will be to bound everything in terms of block-off-diagonal elements of $H$ and its derivatives, i.e., terms of the form $\|PHQ\|$, $\|PH'Q\|$, etc.

\section{Bounds via the resolvent formalism}
\label{ResSec}
Some of the material in this section closely follows Jansen \textit{et al.} (JRS)~\cite{Jansen:07}, adjusted for clarity for our purposes.
We start from the well-known resolvent formula, and then develop various intermediate bounds we need for the final result.

\subsection{Twiddled operators}

If $\Gamma$ is a positively oriented loop in the complex plane encircling the spectrum associated with an orthogonal eigenprojection $P$ of a Hermitian operator $H$, then~\cite{Reed-Simon:book4}: 
\beq
P = \frac{i}{2\pi } \oint_\Gamma (H - z)^{-1} dz ,
\eeq
where $(H-z)^{-1}$ is known as the resolvent.

Using this, it was shown in Lemma 2 of JRS~\cite{Jansen:07} that for every operator $X$ there is a solution $\tilde{X}$ to Eq.~\eqref{eq:tilde} if the eigenvalues in $P$ are separated by a gap in $H$. This solution is written in terms of contour integrals involving the double resolvent:\footnote{Eq.~\eqref{tildeSol} is (up to a minus sign) how the twiddle operation was originally defined in Ref.~\cite[Eq.~(2.11)]{Avron:87}.}
\begin{equation}
    \tilde{X} = \frac{1}{2\pi i} \oint_\Gamma (H - z)^{-1} X (H-z)^{-1}dz = -[({X}^\dag)^\sim]^\dag,
    \label{tildeSol}
\end{equation}
where the contour $\Gamma$ again encircles the portion of the spectrum within $P$. Here $ \tilde{X}$ is block-off-diagonal. The twiddle operation was introduced in Ref.~\cite{Avron:87}, where it was defined via Eq.~\eqref{tildeSol}.

Note that since $P$ and $Q$ both commute with $H$, we can move both $P$ and $Q$ under the twiddle sign, i.e., using Eq.~\eqref{tildeSol} we have 
\bes
\label{eq:51}
\begin{align}
P\tilde{X} &= (PX)^\sim \ , \quad Q\tilde{X} = (QX)^\sim \ , \quad \tilde{X}P = (XP)^\sim \ , \quad \tilde{X}Q = (XQ)^\sim \\
P\tilde{X}Q &= (PXQ)^\sim \ , \quad Q\tilde{X}P = (QXP)^\sim .
\end{align}
\ees
Also note that $\tilde{X}$ is block-off-diagonal~\cite{Jansen:07}, i.e.: 
\bes
\label{eq:24}
\begin{align}
& P\tilde{X}P = Q\tilde{X}Q = 0 \\
& P\tilde{X} = P\tilde{X}Q = \tilde{X}Q\ ,\qquad  Q\tilde{X} = Q\tilde{X}P = \tilde{X}P . 
\label{eq:24b}
\end{align}
\ees

\subsection{Bound on $P'$} 

By definition, $[P,H] = 0$. Differentiating, we obtain:
\begin{equation}
    [H',P] = [P',H] .
\end{equation}
We also know that $P'$ is block-off-diagonal, so by definition [Eq.~\eqref{eq:tilde}] 
\beq
P' = -{H'}^\sim .
\label{eq:P'}
\eeq 
But the tilde operation only depends on the block-off-diagonal elements of $H'$, so that 
\beq
P' = -(PH'Q + QH'P)^\sim ,
\label{eq:P'2}
\eeq 
which implies that as long as this quantity is bounded, $P'$ is as well:
$\|P'\| = \|(PH'Q + QH'P)^\sim\|$.

\subsection{Bound on $\tilde{X}$} 

Suppose that the spectrum of $H(s)$ (its eigenvalues $\{E_i(s)\}$) restricted to $P(s)$ consists of $d(s)$ eigenvalues (each possibly degenerate, crossing permitted) separated by a gap $2\D(s)$ from the rest of the spectrum of $H(s)$. I.e., $d(s)\leq d$, the dimension of the low energy subspace. Under these assumptions JRS proved the following bound in their Lemma~7: 
\begin{equation}
    \|\tilde{X}(s)\| \leq \frac{\sqrt{d(s)}}{\D(s)}\|X\| .
    \label{eq:tildeXbound1}
\end{equation}

We will also use an alternative bound that did not appear in~\cite{Jansen:07}. 
We start with:
\begin{equation}
    \|(H(s)-z)^{-1}\| = \max_i\frac{1}{|E_i(s) -z|} \leq \frac{1}{\D(s)}
    \label{eq:35}
\end{equation}
for $z$ on the contour $\Gamma$ in Eq.~\eqref{tildeSol}, illustrated in Fig.~\ref{fig:contour}. This contour 
is of length $2r(s) + 2\pi \D(s)$ where $r$ is the spectral diameter of $P\mc{H}$ w.r.t $H$. Since $P(s)$ is a spectrum projector, $P\mc{H}$ has a basis of eigenvectors of $H(s)$ with eigenvalues $\lambda_i^{P}$, and we can define:
\beq
r(s) = \max_{\{ |\phi_{{\min}}\rangle, |\phi_{\max}\rangle: \| |\phi\rangle\| =1, P|\phi\rangle = |\phi\rangle \}} (\langle \phi_{\max}|H(s)|\phi_{\max}\rangle -\langle \phi_{{\min}}|H(s)|\phi_{{\min}}\rangle) =  [\max_i\lambda_i^P - {\min}_i\lambda_i^P] ,
\eeq  

\begin{figure}
\centering
\includegraphics[width=0.3\columnwidth]{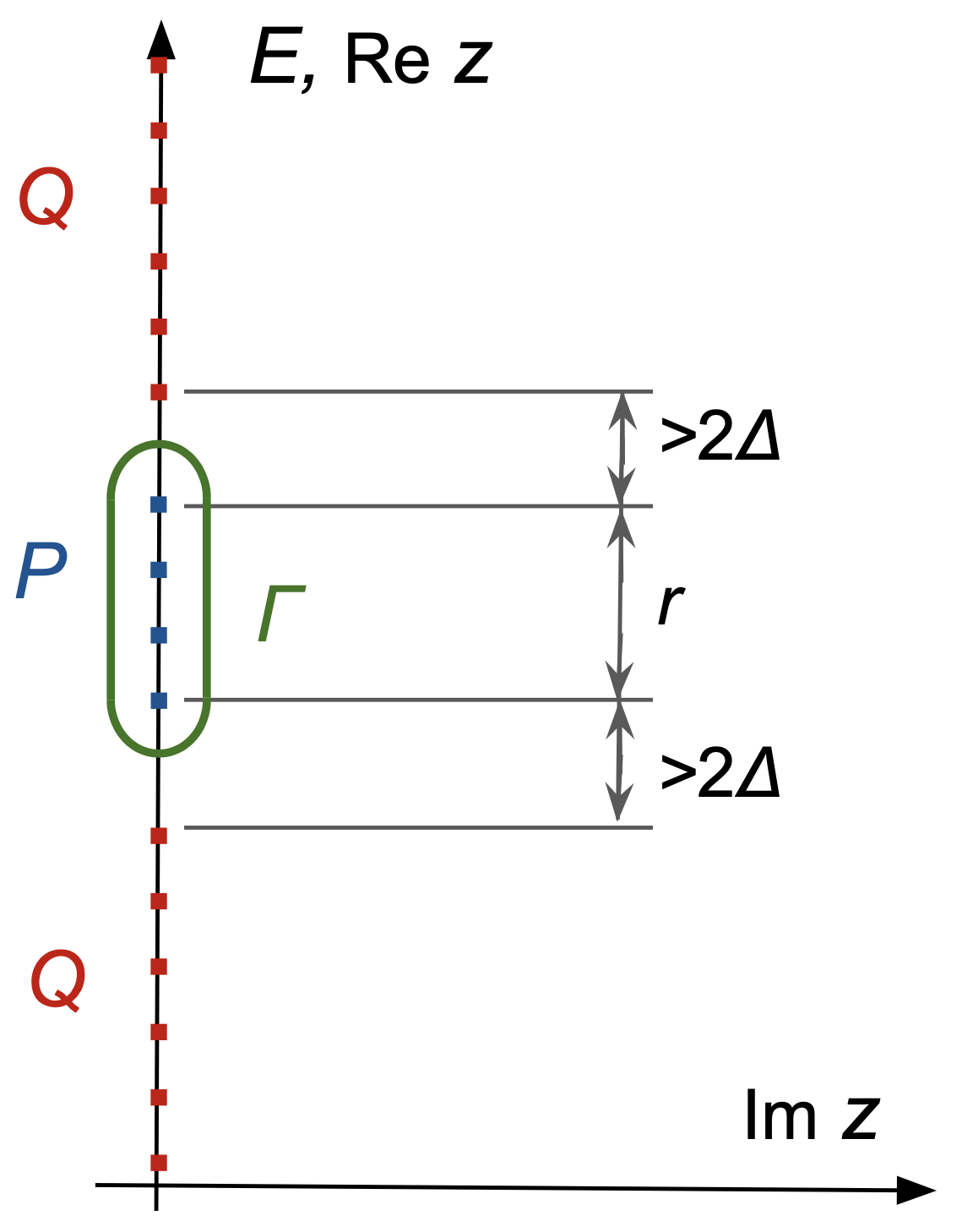}
\caption{An illustration of the integration contour and the various quantities that appear in the statement of Theorem~\ref{th:AT}.}
\label{fig:contour}
\end{figure}

So, bounding the solution $\tilde{X}(s)$ from Eq.~\eqref{tildeSol} directly results in:
\begin{equation}
    \|\tilde{X}(s)\| \leq \frac{2r(s)+ 2\pi \D(s)}{2\pi \D^2(s)}\|X\| .
    \label{eq:tildeXbound2}
\end{equation}
This new bound can be tighter than Eq.~\eqref{eq:tildeXbound1} because it does not depend on $d$,  though this can be offset by $\D$ and $r$.

As stated in Theorem~\ref{th:AT}, we define $\tau$ via Eq.~\eqref{eq:tau} and combine the bounds~\eqref{eq:tildeXbound1} and~\eqref{eq:tildeXbound2} to write
\begin{equation}
     \|\tilde{X}(s)\| \leq \tau(s)\|X\| .
     \label{timescale}
\end{equation}
Here, $\tau$ roughly means the adiabatic timescale. The bound~\eqref{timescale} can be seen as one of the main reasons for introducing the twiddle operation. We will use it repeatedly below. We will omit the $s$-dependence of $\tau$ and $\tilde{X}$ whenever possible in what follows. Note that if $Y$ is any operator that commutes with $H$ then by Eq.~\eqref{tildeSol} we have $\tilde{X}Y = (XY)^\sim, ~  Y\tilde{X} =(YX)^\sim$. Therefore:
\begin{equation}
     \|\tilde{X}Y\|\leq \tau\|XY\|,~ ~ \|Y\tilde{X}\|\leq \tau\|YX\|  \quad \text{if}\quad [Y,H]=0.
     \label{timescale2}
\end{equation}

Likewise, using Eqs.~\eqref{eq:51},~\eqref{eq:24}, and~\eqref{timescale} we can remove a twiddle under the operator norm for the price of a factor of $\tau$ while inserting $P$ and $Q$ at will:
\begin{align}
\label{eq:66}
\|P\tilde{X}\| &= \|\tilde{X}Q\| = \|P\tilde{X}Q\| = \|(P{X}Q)^\sim\| \leq \tau \|PXQ\| .
\end{align}

\subsection{Combining everything into the diabatic evolution bound}

We now combine the various intermediate results above to bound the r.h.s. of Eq.~\eqref{eq:deb1}.

Together with $\|\tilde{X}\| \leq \tau\|X\|$ [Eq.~\eqref{timescale}], Eq.~\eqref{eq:51} yields $\|P(s){P'}^\sim (s)\| \leq \tau \|P(s){P'}(s)\|$. Thus, Eq.~\eqref{eq:deb1} becomes:
\begin{equation}
    \| f(s^*)-P_0 \| \leq \frac{1}{t_f} \left(\tau(0)\|P_0P'(0)\| +\tau(s^*)\|P(s^*)P'(s^*)\| +\int_0^{s^*} \|PP'{P'}^\sim\| + \|P{{P'}^\sim }'\|ds \right) .
\end{equation}
Now, using $[P,H]=0$ and $PP'P=0$, note that: 
\begin{align}
PP'{P'}^\sim P &= PP' \frac{1}{2\pi i} \oint_\Gamma (H - z)^{-1} P' P (H-z)^{-1}dz = PP' \frac{1}{2\pi i} \oint_\Gamma (H - z)^{-1} (P'-PP') (H-z)^{-1}dz \\
&= PP' \frac{1}{2\pi i} \oint_\Gamma (H - z)^{-1} P' (H-z)^{-1}dz = PP'{P'}^\sim .
\end{align}
Also, 
$\|PP'\| = \|(PP')^\dag\| = \|P'P\|$ (since $P$ and $P'$ are Hermitian), so that using Eq.~\eqref{eq:51} we get:
\beq
\|PP{P'}^\sim P\| = \|PP'(P'P)^\sim\| \leq \|PP'\|\|(P'P)^\sim\| \leq \|PP'\|(\tau\|P'P\|) = \tau \|PP'\|^2 .
\eeq
Thus
\beq
    \theta =  \left(\tau(0)\|P_0P'(0)\| +\tau(s^*)\|P(s^*)P'(s^*)\| +\int_0^{s^*} \tau\|PP'\|^2 + \|P{{P'}^\sim }'\|ds \right) .
\eeq

We multiply Eq.~\eqref{eq:P'2} from the left by $P$ to give 
\beq
PP' = -P(PH'Q+QH'P)^\sim = -(PH'Q)^\sim ,
\label{eq:60}
\eeq 
where we used Eq.~\eqref{eq:51}. Therefore, using $\|\tilde{X}\| \leq \tau\|X\|$ again, we find:
\bes
    \label{eq:deb2}
\begin{align}
   & \| f(s^*)-P_0 \| \leq \frac{\theta}{t_f} \\
   & \theta = \tau^2(0)\|P_0H'(0)Q_0\| +\tau^2(s^*)\|P(s^*)H'(s^*)Q(s^*)\| + \int_0^{s^*}
    (\tau^3\|PH'Q\|^2 + \|P{{P'}^\sim }'\|)ds  .
\end{align}
\ees
We have nearly achieved the goal of expressing the diabatic evolution bound in terms of block-off-diagonal elements of $H$ and its derivatives. The last term is not yet in this form and will require the development of additional tools, which we do next.

\subsection{Derivative of the resolvent formula} \label{ResAss}

To take derivatives of the twiddled expressions we need to differentiate the resolvent $R(z,s) = (H(s)-z)^{-1}$. 
By differentiating the identity $(H(s)-z))R(z,s) = I$ we obtain 
\begin{equation}
    \frac{\partial}{\partial s} R(z,s) = -R(z,s)H'(s)R(z,s) .
    \label{eq:41}
\end{equation}

We will apply the derivative formula to our derivation. For example, using Eq.~\eqref{tildeSol} we obtain 
\beq
{P'}^\sim = \frac{1}{2\pi i} \oint_\Gamma (H - z)^{-1} P' (H-z)^{-1}dz
\label{eq:P'sim}
\eeq 
and hence taking the derivative results in 
\begin{align}
  {{P'}^\sim}'  =    \frac{1}{2\pi i} \oint (H - z)^{-1} [-H'(H - z)^{-1}P' +  P'' - P'(H - z)^{-1}H' ] (H-z)^{-1}dz .
  \label{step1}
\end{align}

To bound this expression, we need to prove one more fact.

\subsection{Fact about a triple resolvent} 

We will need to analyze expressions of the form
\begin{equation}
 F(A,B) = \frac{1}{2\pi i} \oint (H - z)^{-1} A(H - z)^{-1}B (H-z)^{-1}dz ,
\end{equation}
which we will use with $A,B = H'$ for the norm of $P''$ and $A,B = H',P'$ for the bound on ${{P'}^\sim}'$ above. I.e., 
\beq
{{P'}^\sim}'  = -F(H',P') + {P''}^\sim -F(P',H')  .
\label{eq:P'tilde'}
\eeq

JRS proved a bound on $F(A,B)$. Since $F(A,B)$ has both diagonal and off-diagonal blocks, they found the bound for each block. We review their proof below, starting from a useful expression for the triple resolvent.

Consider the commutator with the Hamiltonian:
 \begin{align}
     [H,F(A,B)] =\frac{1}{2\pi i} \oint [H-z,(H - z)^{-1} A(H - z)^{-1}B (H-z)^{-1}]dz = A\tilde{B} - \tilde{A}B ,
 \end{align}
where we have inserted $z$ since it is not an operator and therefore commutes with the other term,
and where the second equality follows from Eq.~\eqref{tildeSol}.
 
 Let us denote the off-diagonal block projection by $o(X) =PXQ +QXP=[P,(P-Q)X]$. Note that $P$ and $Q$ commute with $H$, so when we apply $[P,(P-Q)~ \cdot]$ to both sides of the above equation, we get, after some simple algebra:
 \begin{equation}
     [H,o(F(A,B))] = [-(P-Q)(A\tilde{B} - \tilde{A}B),P] .
     \label{eq:44}
 \end{equation}
Now we can apply the definition of the twiddle operation, $[H,\tilde{X}] = [X,P]$ [with $X=-(P-Q)(A\tilde{B} - \tilde{A}B)$], to Eq.~\eqref{eq:44}. It follows that 
\beq
o[F(A,B)] =-\{(P-Q)(A\tilde{B} - \tilde{A}B)\}^\sim .
\label{eq:45}
\eeq

\begin{mylemma}
Multiplication by $(P-Q)$ commutes with the twiddle operation, i.e., $ \{(P-Q)X\}^\sim = (P-Q)\tilde{X}$. 
\end{mylemma}

\begin{proof}
To prove this statement we need to show that $Y=\{(P-Q)X\}^\sim$ and $Y' = (P-Q)\tilde{X}$ satisfy the same defining equation and are both block-off-diagonal. The defining equation of the first is $[H,Y] = [(P-Q)X,P] = (P-Q)XP-PX$. As for the second, note that if we multiply $[H,\tilde{X}]= [X,P]$ by $(P-Q)$ then, since $H$ commutes with $P-Q$, we obtain $[H,Y']=(P-Q)[X,P] = (P-Q)XP-PX = [H,Y]$. Thus $Y'$ satisfies the same defining equation as $Y$. Moreover, by Eq.~\eqref{eq:tilde} $Y=\{(P-Q)X\}^\sim$ is a block off-diagonal operator, and so is $\tilde{X}$, so that $(P-Q)\tilde{X} $ is thus also block-off-diagonal.
\end{proof}

Thus, by Eq.~\eqref{eq:45}, 
\begin{equation}
o[F(A,B)] = -(P-Q)(A\tilde{B} - \tilde{A}B)^{\sim}
\label{eq:46}
 \end{equation}
 
For the block-diagonal part, we need to apply a different strategy. By pole integrations identical to those in~\cite{Jansen:07}, 
which only require that there be a finite number of eigenvalues inside the low energy subspace, we can prove that
 \begin{equation}
F(A,B)-o[F(A,B)] = (P-Q)\tilde{A}\tilde{B}
 \end{equation}
Combining the last two results, we finally obtain (the same as Eq.~(13) in~\cite{Jansen:07}):
  \begin{equation}
F(A,B) = (P-Q)[\tilde{A}\tilde{B} - (A\tilde{B} - \tilde{A}B)^{\sim}] .
\label{step2}
 \end{equation}

Now, using Eqs.~\eqref{tildeSol},~\eqref{eq:P'} and~\eqref{eq:41}, we can express $P''$ as: 
 \begin{equation}
  P''=-{H'^\sim}'  =  \frac{1}{\pi i} \int (H - z)^{-1} H'(H - z)^{-1}H'   (H-z)^{-1}dz - H''^\sim .
\end{equation}
It then follows from Eq.~\eqref{step2} that:
 \begin{equation}
  P''  =  2(P-Q)\{(H'^\sim)^2 -[H',H'^\sim]^\sim\} - H''^\sim 
\label{biggestExpr}
\end{equation}

\subsection{Bounding the last term in the diabatic evolution bound}

We are interested in bounding the last term in Eq.~\eqref{eq:deb2}, which, using Eq.~\eqref{eq:P'tilde'} we can write as:
\begin{equation}
    \|P{{P'}^\sim }'\| = \|P( -F(H',P')- F(P',H') + {P''}^\sim)\| ,
\end{equation}
We now use $F(A,B) = (P-Q)[\tilde{A}\tilde{B} - (A\tilde{B} - \tilde{A}B)^{\sim}]$ [Eq.~\eqref{step2}] to write 
\beq
    \|P{{P'}^\sim }'\| = \|P( -{H'}^\sim {P'}^\sim + (H'{P'}^\sim - {H'}^\sim P')^\sim
    - {P'}^\sim {H'}^\sim + (P'{H'}^\sim - {P'}^\sim H')^\sim + {P''}^\sim)\| .
\eeq
Recall that $P' = -{H'}^\sim$ [Eq.~\eqref{eq:P'}], so that
\beq
\label{eq:62}
    \|P{{P'}^\sim }'\| = \|P( -{H'}^\sim {P'}^\sim + (H'{P'}^\sim)^\sim
    - {P'}^\sim {H'}^\sim - ({P'}^\sim H')^\sim + {P''}^\sim)\| .
\eeq

Repeatedly using the fact that twiddled operators are block-off-diagonal and using Eq.~\eqref{eq:66}:
\begin{align}
\|P{H'}^\sim {P'}^\sim\| &= \|P{H'}^\sim Q {P'}^\sim\| = \|P{H'}^\sim Q {P'}^\sim P\| \leq \| P{H'}^\sim Q\| \|{P'}^\sim P\| \leq \tau^2 \| P{H'} Q\| \|P{P'} \| ,
\end{align}
where in the last inequality we used Eq.~\eqref{tildeSol} and the fact that both $P$ and $P'$ are Hermitian to write $\|{P'}^\sim P\| = \|(P{P'}^\sim)^\dag\| = \|P{P'}^\sim \|$. Similarly:
\begin{align}
\|P(H'{P'}^\sim)^\sim\| &=\|P(H'{P'}^\sim)^\sim Q\| \leq \tau\|PH'({P'}^\sim Q)\| = \tau\|PH'P(P{P'}^\sim)\| \leq \tau \|PH'P\| \|P{P'}^\sim\| \notag \\
&\leq \tau^2 \|PH'P\| \|P P'\| ,
\end{align}
where in the second equality we used $P\tilde{X} = \tilde{X}Q$ [Eq.~\eqref{eq:24b}]. The remaining terms in Eq.~\eqref{eq:62} are similarly bounded:
\bes
\begin{align}
\|P{P'}^\sim {H'}^\sim\| &= \|P{P'}^\sim Q {H'}^\sim\| =  \|P{P'}^\sim Q {H'}^\sim P\| \leq  \|P{P'}^\sim\| \| P {H'}^\sim Q\|  \leq \tau^2 \| P {H'} Q\| \|P{P'}\| \\
\|P({P'}^\sim H')^\sim\| & \leq \tau \|P {P'}^\sim H' \| \\
\|P{P''}^\sim\| &\leq \tau \|P P''\| .
\end{align}
\ees

Combining these bounds yields:
\bes
\label{eq:|PP'sim'|}
\begin{align}
    \|P{{P'}^\sim }'\| &\leq \tau^2(2\|PH'Q\|  +\|PH'P\|)\|PP'\|  + {\tau}(\|P{P'}^\sim H'\| + \| PP''\|)  \\
     &\leq \tau^3(2\|PH'Q\| +\|PH'P\|)\|PH'Q\|  +{\tau}(\|P{P'}^\sim H'\|  + \| PP''\|) ,
\end{align}
\ees
where in the second line we used $\|PP'\|= \|P{H'}^\sim\| = \|P{H'}^\sim Q\|\leq \tau \|PH'Q\|$.

Finally, we use Eq.~\eqref{biggestExpr} for $P''$:
\bes
\label{eq:|PP''|}
\begin{align}
    \|PP''\| &= \| 2P(H'^\sim)^2 -2P[H',H'^\sim]^\sim - PH''^\sim\| \\
    &= 2\|P {H'}^\sim Q {H'}^\sim P\| +2\|P(H'{H'}^\sim)^\sim Q\| + 2\|P({H'}^\sim H')^\sim\| + \|PH''^\sim\| \\
    &\leq 2\| P {H'}^\sim Q\| \| Q {H'}^\sim P\| + 2\tau\| PH'{H'}^\sim Q\| +2\tau\|P{H'}^\sim H'\| + \|PH''^\sim\| \\
    &\leq 2\tau^2(\|PH'Q\|  + \|PH'P\|) \|PH'Q\| +2 \tau \|PH'^\sim H'\| +\tau\|P{H''}Q\| 
\end{align}
\ees

To deal with the two terms that still contain $\sim$ ($\|P{P'}^\sim H'\| $ and $\|PH'^\sim H'\|$), we have no choice but to use the constants $c_0, c_1$ introduced in Sec. \ref{sec:result}:
\begin{equation}
    {H'}^2 \leq c_0 +c_1  H^{2} .
    \label{eq:assum0}
\end{equation}

We use this assumption as follows. First, it implies that $P{H'}^\sim {H'}^2 {H'}^\sim  P \leq \sum_{k=0}^1c_kP{H'}^\sim {H}^{2k} {H'}^\sim P$. Hence, upon taking norms of both sides:
\begin{align}
    \|P{H'}^\sim H'\|^2 &= \|P{H'}^\sim {H'}^2 {H'}^\sim  P\| \leq \sum_{k=0}^1 c_k\|P{H'}^\sim {H}^{2k} {H'}^\sim  P\| = \sum_{k=0}^1 c_k\|P{H'}^\sim {H}^k\|^2 = \sum_{k=0}^1 c_k\|P(H' {H}^k)^\sim\|^2 \notag \\
&    \leq \sum_{k=0}^1 c_k\tau^2\|PH'H^k Q\|^2, 
\end{align}
where in the first equality we used $\|A\|^2 = \|AA^\dag\|$ and in the last equality we made use of $\tilde{X}Y = (XY)^\sim$ when $[Y,H]=0$, and then applied Eq.~\eqref{eq:66}. 

Similarly, using $P' = -{H'}^\sim$:
\begin{align}    
    \|P{P'}^\sim H'\|^2 &=\|P{H'}^{\sim\sim}H'\|^2 = \|P{H'}^{\sim\sim}{H'}^2 {H'}^{\sim\sim} P\| \leq \sum_{k=0}^1 c_k\|P{H'}^{\sim\sim}{H}^{2k} {H'}^{\sim\sim} P\| \notag \\ 
    &= \sum_{k=0}^1 c_k\|P{H'}^{\sim\sim}{H}^k\|^2 \leq \sum_{k=0}^1 c_k\tau^4\|PH'H^kQ\|^2 .
\end{align}

The quantity $\|PH'HQ\|$ appearing for $k=1$ is usually well-behaved with $\Lambda$ as we will see in examples in Sec.~\ref{sec:fluxQ}. In case it is not, we need to take a step back and recall that we obtained it via the bound $\|P(H'H)^\sim Q\| \le \tau \|PH'HQ\|$, which follows from Eq.~\eqref{eq:66}. We thus consider undoing this bound and replacing 
$\tau\|PH'HQ\| \to \|P(H'H)^\sim Q\|$. Using the definition of the $\sim$ operation [Eq.~\eqref{tildeSol}]:
\bes
\begin{align}    
P(H'H)^\sim Q &=  P\frac{1}{2\pi i} \oint_\Gamma (H - z)^{-1} H'(H-z+z) (H-z)^{-1}Qdz \\
&= PH' Q + P\frac{1}{2\pi i} \oint_\Gamma z(H - z)^{-1} H' (H-z)^{-1}Qdz ,
\end{align}
\ees
where to obtain the second equality we used $P\frac{1}{2\pi i} \oint_\Gamma (H - z)^{-1}dz H' Q = PPH' Q$.

The choice of zero energy right in the middle of the eigenvalues corresponding to $P\mathcal{H}$ ensures that $|z|\leq r/2+\Delta$ for $z\in \Gamma$ (see Fig.~\ref{fig:contour}). Using this fact along with Eq.~\eqref{eq:tildeXbound2} then results in the bound
\begin{equation}
   \|P(H'H)^\sim Q \| \leq 
   \|PH' Q \|( 1 + \tau_{\text{new}}(r/2+\Delta)) , \qquad \tau_{\text{new}} \equiv \frac{2r +2\pi \Delta}{2 \pi \Delta^2} .
   \label{eq:ourbound}
\end{equation}
Alternatively, a slight adjustment to the derivation in Ref.~\cite{Jansen:07} gives:
\begin{equation}
   \|P(H'H)^\sim Q \| \leq  \|PH' Q \|( 1 + \tau_{\text{JRS}}r/2), \qquad \tau_{\text{JRS}} \equiv \frac{\sqrt{d(s)}}{\Delta(s)} .
      \label{eq:JRSbound}
\end{equation}
Combining Eqs.~\eqref{eq:ourbound} and~\eqref{eq:JRSbound} we obtain an alternative form for our bound
\begin{equation}
    \tau\|PH'HQ\| ~ \to ~ \|PH' Q \|( 1 + \text{min}(\tau_{\text{new}}(r/2+\Delta),\tau_{\text{JRS}}r/2)) .
\end{equation}

Collecting all these bounds into Eqs.~\eqref{eq:|PP'sim'|} and~\eqref{eq:|PP''|}, we obtain:
\bes
\label{eq:78}
\begin{align}
    \|P{{P'}^\sim }'\|  &\leq \tau^3(2\|PH'Q\| +\|PH'P\|)\|PH'Q\| +  \sqrt{\sum_k c_k\|PH'H^kQ\|^2}  +  2\|PH'Q\|^2) \\ 
    &\quad +2 \tau^3 (\|PH'P\| \|PH'Q\| +\sqrt{\sum_k c_k\|PH'H^kQ\|^2}) +\tau^2\|P{H''}Q\|  \\ 
    &=\tau^3(4\|PH'Q\| + 3 \|PH'P\|)\|PH'Q\| + \tau^2\|PH''Q\| + 3\tau^3\sqrt{\sum_k c_k\|PH'H^kQ\|^2}
\label{eq:78c}
\end{align}
\ees
We are now ready to write down the diabatic evolution bound in its final form, by combining Eqs.~\eqref{eq:diff},~\eqref{eq:18},~\eqref{eq:deb2}, and~\eqref{eq:78}:
\bes
\label{eq:deb3}
\begin{align}
   & \|[\Uad (s^*) -  \Utot(s^*)]P_0\| \leq  \frac{\theta}{t_f}\\
    & \quad \theta =  \tau^2(0)\|P_0H'(0)Q_0\| +\tau^2(s^*)\|P(s^*)H'(s^*)Q(s^*)\| + \int_0^{s^*}ds [\tau^3(5\|PH'Q\| + 3 \|PH'P\|)\|PH'Q\|  \notag \\
    &\qquad \qquad + \tau^2\|PH''Q\| + 3\tau^3\sqrt{\sum_k c_k\|PH'H^kQ\|^2}]   ,
    \label{eq:deb3-b}
\end{align}
\ees 
where the expression for $\theta$ coincides with the one in Eq.~\eqref{adtime}, and hence serves as the end of the proof of Theorem~\ref{th:AT}. It is worth recalling here also that $\tau$ contains a gap dependence via Eq.~\eqref{eq:tau}.

Note that despite appearances due to the block-off-diagonal form of this bound, all of the terms involved can be bounded by norms of some $d_P\times d_P$ matrices ($d_P=\text{rank}(P)$):
\begin{equation}
    \|PH'Q\| \leq \sqrt{\|P{H'}^2 P \|}, \quad  \|PH'HQ\| \leq  \sqrt{\|PH'H^{2} H' P \|} ,
    \label{eq:93}
\end{equation}
where the inequalities follow by writing (for any Hermitian operator $A$): $\|PAQ\|= \text{max}_{|v\rangle, |w\rangle} \langle v|PAQ|w\rangle \leq  \text{max}_{|v\rangle, |w\rangle} ($ ~\quad\quad ~ ~ ~$ \langle v|PA|w\rangle) =\|PA\|$, and $\|PA\|^2= \|PA(PA)^\dag\| \leq \|PA^2 P\|$, so that $\|PAQ\|^2 \leq \|PA^2 P\|$. 

Before we proceed, let us comment briefly on a physical consequence of the bound  $\|[\Uad (s^*) -  \Utot(s^*)]P_0\| \leq  \frac{\theta}{t_f}$ that we have just proven [Eq.~\eqref{eq:deb3}]. In Sec.~\ref{Obound} we gave a bound on the difference in expectation value of an observable $O$ between the exact and the adiabatic evolution. Suppose that $O$ is a unit norm observable such as the Pauli matrix $\sigma^z\equiv Z$ or $\sigma^x\equiv X$;  measuring $Z$ on a single qubit in an $n$-qubit system is a standard ``computational basis" measurement.  For this example, Eq.~\eqref{eq:Obound} then becomes:
\begin{equation}
   \langle \phi| \Utot^\dag  Z\Utot|\phi\rangle -\langle \phi| \Uad ^{\dag}  Z\Uad |\phi\rangle  \leq \frac{2\theta}{t_f}  .
\end{equation}
This means that a measurement of $Z$ at $t_f$ has an expectation value that -- provided $\frac{\theta}{t_f}\ll 1$ -- is well described by an expectation value computed from the evolution $U_{\text{ad}}$ that never leaves the low-energy subspace, which is the qubit subspace. The error between the two is given by the bound above. In Sec.~\ref{eHam} we discuss the effective Hamiltonian (a qubit Hamiltonian for this example) generating this approximate evolution in more detail, with the aim of providing a recipe for numerical simulations of qubit Hamiltonians that can predict the outcomes of superconducting circuit experiments.

\section{Examples}
\label{sec:fluxQ}

We consider examples motivated by adiabatic quantum computing and quantum annealing with flux qubits~\cite{Dwave,Harris:2010kx,Weber:2017aa,Novikov:2018aa,khezri2020annealpath}. We first discuss inductively coupled flux qubits in terms of generic circuit Hamiltonians. We use Theorem~\ref{th:AT} to derive general bounds on the deviation between the actual evolution described by these circuit Hamiltonians and the evolution in the desired low-energy subspace defined by $P$. Next we discuss specific models of single flux qubits, for which we can explicitly exhibit the dependence of our bounds on the circuit parameters.

\subsection{Application to coupled flux qubits} 

\begin{figure}
\centering
\includegraphics[width=0.5\columnwidth]{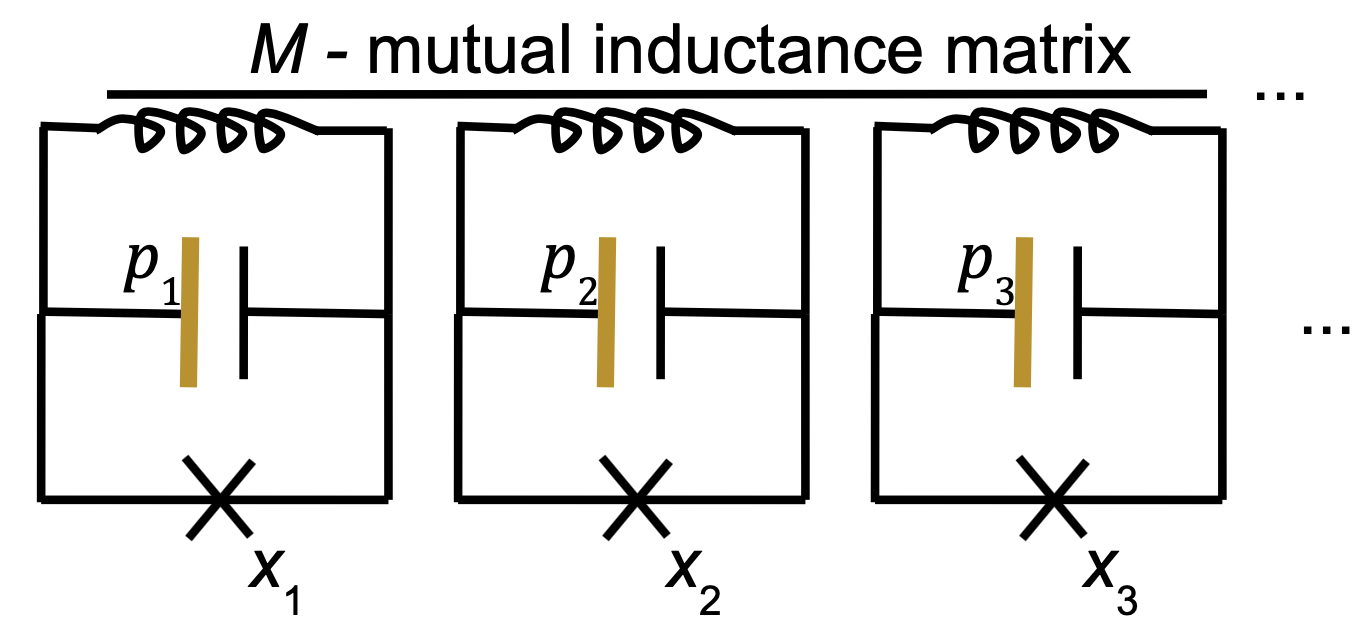}
\caption{The circuit corresponding to Eq. (\ref{eq:91}). The horizontal line above the inductors indicates that each pair is coupled via a mutual inductance $M_{ij}$, and the self inductance is the diagonal $M_{ii}$}
\label{fig:dWc}
\end{figure}

An interesting example is the circuit Hamiltonian describing inductively coupled superconducting flux qubits~\cite{Orlando:99}:
\begin{equation}
\label{eq:91}
    H_{\text{flux}}(s) = \sum_i \hat{p}_i^2 + B_i(s) \cos (\hat{x}_i + \varphi_i(s))  + \sum_{ij}M_{ij}(s)\hat{x}_i\hat{x}_j  ,
\end{equation}
where $\hat{p}_i$ and $\hat{x}_i$ are canonically conjugate momentum and position operators, respectively. The remaining quantities are scalar control parameters: $\varphi_i$ are  control fluxes, $M_{ij}$ are matrix elements of the mutual inductance matrix, and $B_i$ are barrier heights depending on more control fluxes~\cite{Wendin:2017aa}. A simplified circuit described by this equation is shown in Fig.~\ref{fig:dWc}. For notational simplicity we drop the hat (operator) notation below.

The Hamiltonian $H_{\text{flux}}(s)$ is defined over an infinite-dimensional Hilbert space and is unbounded: $\|H_{\text{flux}}(s)\| =\text{max}_{|v\rangle}\langle v|H_{\text{flux}}(s)|v\rangle$  is infinite for $|v\rangle$ maximized over a typical Hilbert space. One such space can be defined by choosing
\begin{equation}
    H_0 =  \sum_i p_i^2 +M_{ii}(0)x_i^2 ,
\end{equation}
and considering eigenvectors $|v\rangle =\otimes_i |n_i\rangle$ of this collection of harmonic oscillators. Clearly in some contexts in physics arbitrarily high $n_i$ will appear as a physical state, which would lead to arbitrarily large $\langle v|p_i^2|v\rangle$, $\langle v|x_i^2|v\rangle$, $\langle v|H_0|v\rangle$,  and $\langle v|H_{\text{flux}}(s)|v\rangle$. Indeed the operators involved would normally be referred to as unbounded. We note that in the definition of the norm  $\|\cdot\|_{\mathcal{L}(\mathcal{D},\mathcal{H})}$\cite{Teufel:book} discussed in Sec.~\ref{graphNormHere}, these operators are bounded with respect to the Hamiltonian. We choose instead to impose a cutoff on the Hamiltonian directly. This allows us to compare with the JRS result which requires a finite-dimensional Hamiltonian. Consider a projector $P_\Lambda$ on states with all $n_i\leq \Lambda$, and for any operator $O$ on the original infinite-dimensional Hilbert space  define $O^\Lambda$ as the finite-dimensional matrix that is the $P_\Lambda$ block of $P_\Lambda O^\Lambda P_\Lambda$. Now using the standard definition of the norm for finite-dimensional matrices we can get $\|p_i^\Lambda\| = \Theta(\sqrt{\Lambda}),~,\|x_i^\Lambda\| = \Theta(\sqrt{\Lambda}),~\|H_0^\Lambda\| = \Theta(\Lambda),~ \|H_{\text{flux}}^\Lambda(s)\| = \Theta(\Lambda)$. Below we will omit the superscript $\Lambda$, but all the expressions that follow are understood to hold in this finite-dimensional space.

\subsubsection{Constant mutual inductance matrix}

We first consider the case $M_{ij}(s)=M_{ij}$.
As we shall see, in this case $\|H'\|$ does not grow with the cutoff, $H'^2 \leq c_0$ is sufficient, and previously developed bounds such as JRS's will not depend on the cutoff either, though recall that by Corollary~\ref{cor:JRScomp} we can obtain a tighter bound.

The derivative is:
\begin{equation}
    H_{\text{flux}}'(s) = \sum_i  B_i'(s) \cos (x_i + \varphi_i(s))   -B_i(s)\varphi_i'(s) \sin (x_i + \varphi_i(s)) , 
\end{equation}
and we note that 
\begin{equation}
    \| H_{\text{flux}}'(s)\|  \leq \sum_i  |B_i'(s)| +B_i(s)|\varphi_i'(s)|  = \sqrt{c_0(s)} ,
\end{equation}
where as long as $B_i(s)$ and $\varphi_i(s)$ are smooth functions of $s$ then $c_0(s)$ is finite, does not depend on the cutoff $\Lambda$ and has dimensions of energy:
\begin{equation}
   c_0(s) =  \left(\sum_i  |B_i'(s)| +B_i(s)|\varphi_i'(s)| \right)^2
\end{equation}

The final error upper bound [Eq.~\eqref{adtime}] simplifies to: 
\bes
\label{eq:95}
\begin{align}
 \theta &= 
  {\tau^2(0)} \|P_0H'(0)Q_0\| +\tau^2(s^*)\|P(s^*)H'(s^*)Q(s^*)\| + \int_0^{s^*}ds [\tau^3(5\|PH'Q\| + 3 \|PH'P\|)\|PH'Q\| \\
        &\qquad + \tau^2\|PH''Q\| + 3\tau^3\sqrt{c_0}\|PH'Q\|]   .
        \end{align}
\ees

Now, since in this example $\|H'(s)\|$ is finite and $\Lambda$-independent $\forall s$, in fact the projection $P$ is not necessary and known bounds are already $\Lambda$-independent. Indeed, the JRS bound for $\theta(s^*)$ quoted in Eq.~\eqref{eq:JRS-AT} 
is clearly $\Lambda$-independent for the present example [recall Corollary~\ref{cor:JRScomp}]. Thus in the next subsection we consider an example where $\|H'(s)\|$ diverges with $\Lambda$.

\subsubsection{Time-dependent mutual inductance matrix}

Generally, to implement a standard adiabatic quantum computing or quantum annealing protocol, the mutual inductance matrix $M_{ij}$ cannot be constant (e.g., see Ref.~\cite{Harris:2010kx}). Thus we consider a second example of a circuit Hamiltonian of superconducting flux qubits, more appropriate for both quantum annealing and our purpose of demonstrating the case of unbounded Hamiltonians with cutoff. Consider the Hamiltonian in Eq.~\eqref{eq:91} and its derivative:
\begin{equation}
    H_{\text{flux}}'(s) = \sum_i  B_i'(s) \cos (x_i + \varphi_i(s))   -B_i(s)\varphi_i'(s) \sin (x_i + \varphi_i(s))  + \sum_{ij}M_{ij}'(s)x_ix_j  .
    \label{eq:97}
\end{equation}
The term $M_{ij}'(s)x_ix_j$, containing the derivative of the time-dependent mutual inductance matrix, now grows arbitrarily large in norm with $\Lambda$ due to the $x_ix_j$ terms (recall that $x_i$ are operators), so that the JRS version of the adiabatic theorem [Eq.~\eqref{eq:JRS-AT}] has an adiabatic timescale that is arbitrarily large in $\Lambda$ and we need to resort to Theorem~\ref{th:AT}. Note that $M_{ij}(s)$ is always a positive matrix. Denote its lowest eigenvalue by $l = {\min}\lambda_M$. Then we can bound:
\begin{equation}
    M \geq lI, \quad \Rightarrow \quad \sum_{ij}M_{ij}(s)x_ix_j \geq l\sum_ix_i^2 .
\end{equation}
Note also that
\begin{equation}
    \|M'\| I \geq M' \quad \Rightarrow \quad  \|M'\|\sum_ix_i^2\geq  \sum_{ij}M_{ij}'(s)x_ix_j ,
\end{equation}
so that we obtain:
\begin{equation}
    \frac{\|M'\|}{l}\sum_{ij}M_{ij}(s)x_ix_j \geq \sum_{ij}M_{ij}'(s)x_ix_j .
\end{equation}
Substituting this inequality into Eq.~\eqref{eq:97} we have: 
\begin{equation}
    H_{\text{flux}}'(s) \leq \sum_i  |B_i'(s)| +B_i(s)|\varphi_i'(s)| +  \frac{\|M'\|}{l}\sum_{ij}M_{ij}(s)x_ix_j .
\end{equation}
We now add a (positive) $p^2$ term and add and subtract the $\cos$ term to complete the Hamiltonian:
\begin{equation}
    H_{\text{flux}}'(s) \leq \sum_i  (|B_i'(s)| +B_i(s)|\varphi_i'(s)|) +  \frac{\|M'\|}{l}H_{\text{flux}}(s) - \frac{\|M'\|}{l} B_i(s) \cos (x_i + \varphi_i(s)) .
\end{equation}
Bounding the last term in the same way as the first two, we obtain:
\begin{equation}
    H_{\text{flux}}'(s) \leq \sum_i  (|B_i'(s)| +B_i(s)|\varphi_i'(s)|) +  \frac{\|M'\|}{l}H_{\text{flux}}(s) + \frac{\|M'\|}{l} \sum_i|B_i(s)| .
\end{equation}
Denote $a_0 = \sum_i  (|B_i'(s)| +B_i(s)|\varphi_i'(s)|) + \frac{\|M'\|}{l} |B_i(s)|  $ and $a_1 =\frac{\|M'\|}{l} $,
then $H_{\text{flux}}' \leq a_0 +a_1H_{\text{flux}}$.
For the square of the derivative, we obtain:
\begin{equation}
    H_{\text{flux}}'^{2} \leq (a_0 + a_1 H_{\text{flux}})^2 \leq  (a_0 + a_1 H_{\text{flux}})^2 +(a_0 - a_1 H_{\text{flux}})^2 \leq 2a_0^2 + 2a_1^2H_{\text{flux}}^2
\end{equation}
Thus the constants we defined in the general notation of Eq.~\eqref{eq:main-assump} are $\sqrt{c_0} =\sqrt{2}a_0$ and $\sqrt{c_1} =\sqrt{2}a_1$, or, explicitly:
\begin{equation}
    \sqrt{c_0} =  \sqrt{2}\sum_i  (|B_i'(s)| +B_i(s)|\varphi_i'(s)|) + \frac{\|M'\|}{l} |B_i(s)|, \quad  \sqrt{c_1} =\sqrt{2}\frac{\|M'\|}{l}
\end{equation}
The final numerator in the diabatic evolution bound [Eq.~\eqref{adtime}] becomes:
\begin{align}
\label{eq:106}
    \theta & = \tau^2(0) \|P_0H'(0)Q_0\| +\tau^2(s^*)\|P(s^*)H'(s^*)Q(s^*)\| + \int_0^{s^*}ds [\tau^3(5\|PH'Q\| + 3 \|PH'P\|)\|PH'Q\| \notag \\
        &\qquad + \tau^2\|PH''Q\| + 3\tau^3\sqrt{c_0\|PH'Q\|^2 +c_1\|PH'HQ\|^2}]    .
        \end{align}
Contrasting this with Eq.~\eqref{eq:95} for the case of a constant mutual inductance matrix, we see that the only differences are the appearance of the new term $c_1\|PH'HQ\|^2$ and an extra contribution from $M_{ij}'$ to every $H'$.

\subsection{Adiabatic timescale via superconducting qubit circuit parameters}
\label{EjEcSec}

The bounds above are stated in terms of the circuit parameters $B_i$ and $M_{ij}$ but are too abstract to be practically useful. In this subsection we consider more specific models and arrive at practically useful bounds which also illustrate the utility of our approach for dealing with unbounded operators with a cutoff.

We consider two types of flux-qubit circuit Hamiltonians:

\bes
\label{eq:H-flux}
\begin{align}
    H_{\text{CJJ}} &= E_C \hat{n}^2  + E_J b \cos \hat{\phi} + E_L (\hat{\phi}- f)^2 , \quad \phi \in [-\infty, \infty],
    \label{eq:H-CJJ}\\
        H_{\text{CSFQ}} &= E_C \hat{n}^2 + E_J b \cos \hat{\phi} - E_\alpha \cos\frac{1}{2}(\hat{\phi}- f)  \quad \phi \in [-2\pi, 2\pi].
    \label{eq:H-CSFQ}
\end{align}
\ees
\begin{figure}
\centering
\includegraphics[width=0.5\columnwidth]{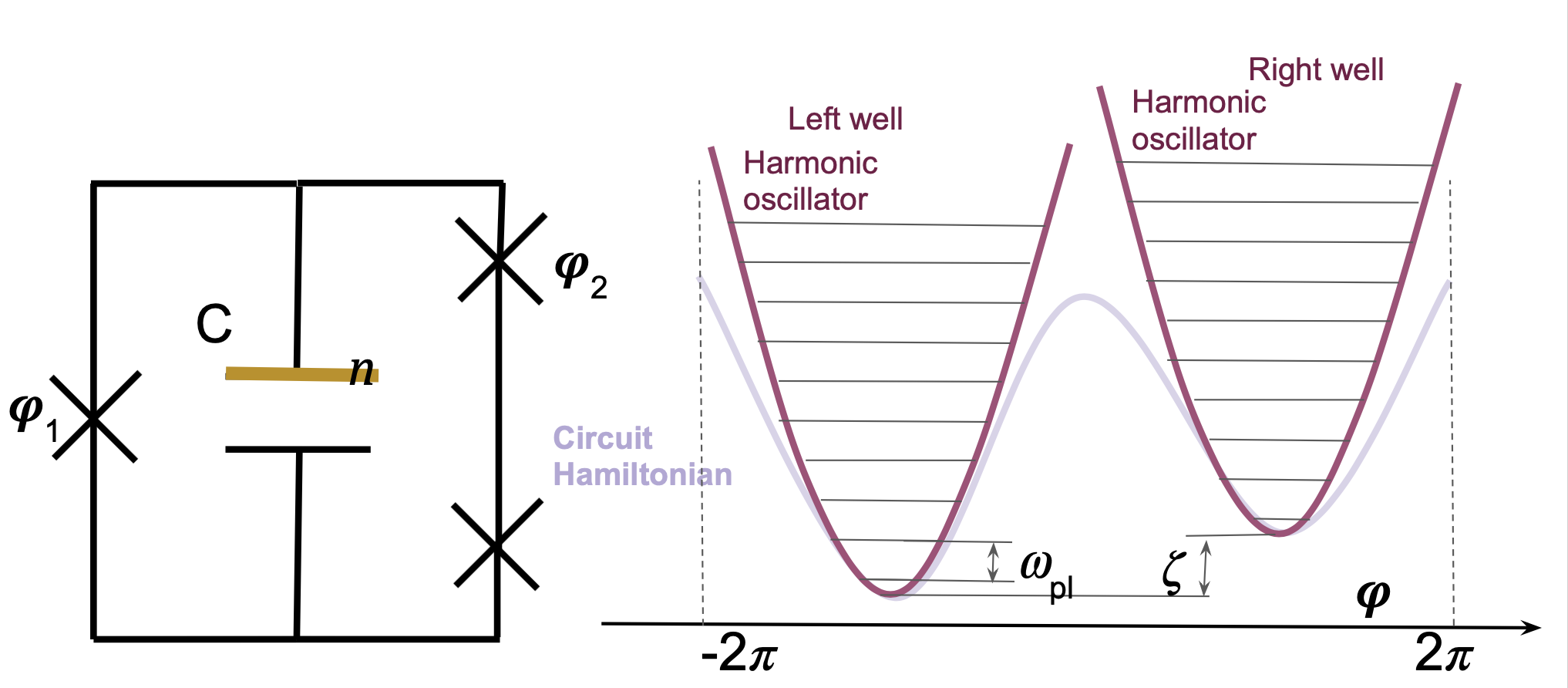}
\caption{The circuit loosely corresponding to Eq. (\ref{eq:H-CSFQ}), and the potential for the phase variable $\phi$. The lowest two wells are approximated as harmonic oscillators, with bias $\zeta$ and tunneling $\xi$ between the groundstates of the wells. The relationship between $\varphi_1,\varphi_2$ and $\phi$ is discussed in Ref.~\cite[Supplementary Material, p.17]{Yan15a}, which also explains how $H_{\text{CSFQ}}$ can be obtained by analyzing the circuit shown here.}
\label{fig:CSFQone}
\end{figure}

As we explain below, $H_{\text{CJJ}}$ describes a compound Josephson junction (CJJ) rf SQUID qubit~\cite{Harris:2008lp}, while $H_{\text{CSFQ}}$ describes a capacitively shunted flux qubit (CSFQ)~\cite{Yan15a}. $H_{\text{CSFQ}}$ can be obtained by analyzing the circuit displayed in Fig.~\ref{fig:CSFQone}.
Note that in the notation of Eq.~\eqref{eq:91}, the canonically conjugate operators $\hat{n}$ (charge stored in the capacitor $C$) and $\hat{\phi}$ (flux threading the circuit) are identified with $\hat{p}$ and $\hat{x}$, respectively, and that in the transmon case $E_L=E_\alpha=0$~\cite{transmon-invention}.\footnote{Note a factor of $4$ difference in the definition of $E_C$ between the latter and our Eq.~\eqref{eq:H-flux}: our definition is $E_C = (2e)^2/2C$, and $H = E_C n^2 +\dots$, while the definition in Ref.~\cite{transmon-invention} is $E_C =e^2/2C$ and $H = 4 E_C n^2+\dots$.} 

The quadratic self-inductance term $E_L (\hat{\phi}- f)^2$ is responsible for the divergence of $\|H'_{\text{CJJ}}\|$ with the cutoff $\Lambda$, just like 
the time-dependent mutual inductance in Eq.~\eqref{eq:91}. Thus, the JRS adiabatic theorem once again provides an unphysical dependence on the cutoff and the bound we derived in Eq.~\eqref{eq:106} can be used instead. The adiabatic timescale depends on the choice of schedules for the controls $b$ and $f$. To illustrate what enters this choice, we first explain how $H_{\text{CJJ}}$ can be reduced to an effective qubit Hamiltonian. We would like to stress that we only need the qubit approximation for the schedule choice; the adiabatic timescale we find is a property beyond the qubit approximation, and the approximation itself is not used anymore after the schedule is set. Before presenting the result for CJJ qubits, we borrow the same set of tools to find the effective qubit Hamiltonian and explicitly compute our bounds for the capacitively shunted flux qubit
described by a simpler Hamiltonian $H_{\text{CSFQ,sin}}$ where we retain just one of the trigonometric terms:
\begin{equation}
     H_{\text{CSFQ,sin}} = E_C \hat{n}^2 + E_J b \cos \hat{\phi} - E_\alpha \sin\frac{\hat{\phi}}{2}\sin \frac{f}{2}\ ,  \qquad \phi \in [-2\pi, 2\pi].
     \label{eq:120}
\end{equation}
Note that the derivatives of $H_{\text{CSFQ}}$ and $H_{\text{CSFQ,sin}}$ do not grow in  norm with the cutoff $\Lambda$, so in this case the JRS adiabatic theorem provides a useful baseline, but as explained below we will obtain a somewhat tighter bound.

The quantities $b\geq 1$ and $f\geq 0$ are time-dependent controls that can be chosen at will. Ideally we would like the effective qubit Hamiltonian (Sec.~\ref{eHam}) to match a desired quantum annealing ``schedule" $\omega_{\text{q}}( (1-s )X + sZ)$ where $s=t/t_f$ is the dimensionless time. However in practice for calibration of the annealing schedule an approximate method for choosing $b(s),f(s)$ is used instead. Here we will also follow this approximate method for simplicity; thus we will not know the true effective qubit-Hamiltonian $H_{\text{eff}}$ the schedule is implementing, but we will be able to accurately bound the error of that qubit description. This is in line with our goal of providing a useful theoretical result to guide current experiments with superconducting circuits: the error would characterize, e.g., the leakage to the non-qubit states for fast anneals. The true effective Hamiltonian $H_{\text{eff}}$, and correspondingly a precise method for choosing $b(s),f(s)$, can be found straightforwardly in a numerical simulation, which we leave for future work.

The approximate method is as follows. 
\begin{mydef}
\label{imprec}
Using the exact circuit description, we compute a $2\times2$ operator $H_\text{q}$ defined as follows. $H_\text{q}$ acts on a $2$-dimensional Hilbert space corresponding to the low-energy subspace of the circuit Hamiltonian. The basis for $H_\text{q}$ in that subspace is chosen to diagonalize the low-energy projection of $\hat\phi$. The energy levels of $H_q$ are chosen to exactly match the two levels of the circuit Hamiltonians, up to a constant shift. Once we obtain the relationship between $b(s),f(s)$ and $H_q$, we find $b(s),f(s)$ by requiring:
\begin{equation}
    H_{\text{q}} =\omega_{\text{q}}( (1-s  +\delta)X + sZ) ,
    \label{eq:Hq}
\end{equation}
where $\delta>0$ is a certain precision parameter we discuss below (ideally $\delta=0$).
\end{mydef}

Note that the true effective Hamiltonian $H_{\text{eff}}$ is isospectral to $H_{\text{q}}$, and is a rotation of $H_{\text{q}}$ to the basis determined by $U_{\text{eff}}$ as will be prescribed in Sec.~\ref{eHam}. In this section we only obtain explicit values of $\theta$ [the timescale in the error bounds~\eqref{eq:thetas}] for an evolution up to $s=s^*$, and demonstrate an improvement (small for CSFQ qubits, diverging as $\Theta(\Lambda)$ for CJJ qubits) relative to the JRS version, which yields
\begin{equation}
    \theta^{\text{JRS}}_{\text{CSFQ}}(s^*) = O\left(\frac{1}{\omega_{\text{pl}}(s^*) (1-s^* +\delta)} \right), \quad \theta^{\text{JRS}}_{\text{CJJ}}(s^*) =\Theta(\Lambda)
   \label{eq:113}
\end{equation}
while our new bound yields
\begin{equation}
   \theta^{\text{new}} = O\left(\frac{1}{\omega_{\text{pl}}(s^*) (1-s^* +\delta) \text{ln}\frac{\omega_{\text{pl}}(s^*) }{ \omega_\text{q}(1-s^* +\delta) }} \right) .
   \label{eq:114}
   \end{equation}
Here the qubit approximation starts at $b(0)=1$ and ends at $b(s^*)>1$. In the Introduction [below Eq.~\eqref{eq:4}] these results were reported for the special case of $s^*=1,$  $b(1) =\mathcal{B}>1$.  The gap $2\Delta(s)$ separating the qubit subspace from the rest of the Hilbert space (recall Fig.~\ref{fig:contour}) will turn out to be well approximated by the plasma frequency $\omega_{\text{pl}}(s) =2\sqrt{E_C E_J b(s)}$. To leading order only the final value of that gap $\omega_{\text{pl}}(s^*) =2\sqrt{E_C E_J b(s^*)}$ enters our bound. These results hold in the relevant regime $E_J/E_C\gg1$, $1-s^* +\delta \ll 1$. The quantities appearing in our result for the adiabatic timescale are illustrated in Fig.~\ref{fig:deldefined}.

\begin{figure}
\centering
\includegraphics[width=0.5\columnwidth]{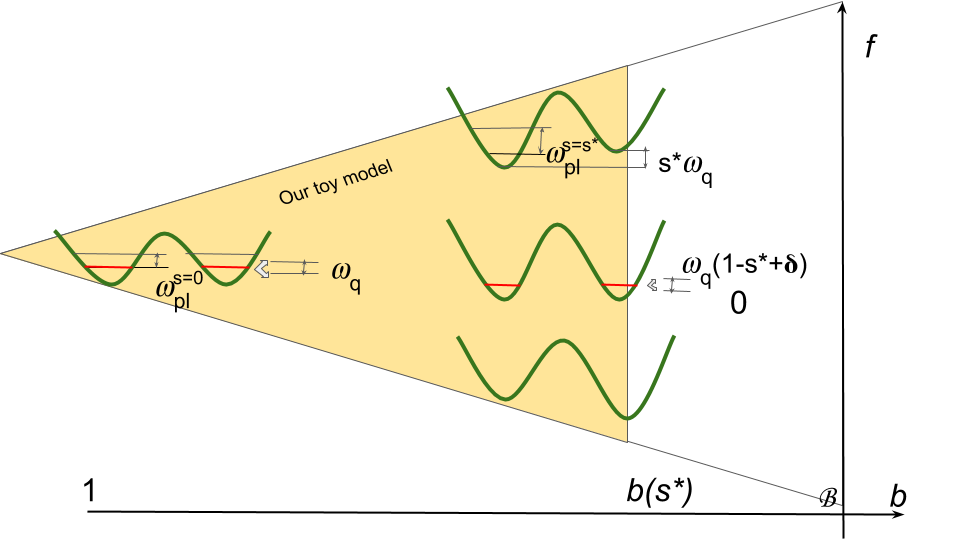}
\caption{For various target Hamiltonians between $+\omega_{\text{q}}Z$ and $-\omega_{\text{q}}Z$ the anneal paths in the parameter space $b(s),f(s)$ occupy the white triangle. The yellow triangle indicates the range of applicability of the qubit approximation for anneals with $t_f \gg \theta(s^*)$. The splittings $\omega_{\text{q}}$ and $\omega_{\text{q}} \delta$ are obtained at zero bias at the beginning and the end of the anneal respectively. Maximum bias also yields $\omega_{\text{q}}$ at the end of the anneal. The plasma frequency $\omega_{\text{pl}}$ is the frequency of each well, and increases throughout the anneal towards the value $\omega_{\text{pl}}(s^*)$ that enters $\theta(s^*)$ in Eq.~\eqref{eq:113} .}
\label{fig:deldefined}
\end{figure}

For notational simplicity we again drop the hat (operator) symbols from now on. The goal of the rather lengthy calculations that follow in the remainder of this section is to assign a physical significance to the various quantities that appear in Eqs.~\eqref{eq:theta-JRS} and~\eqref{eq:theta-new}, expressed in terms of the parameters of CJJ and CSFQ circuits, so as to eventually derive Eqs.~\eqref{eq:113} and~\eqref{eq:114}.

\subsubsection{Compound Josephson junction (CJJ) rf SQUID}

\begin{figure}
\centering
\includegraphics[width=0.5\columnwidth]{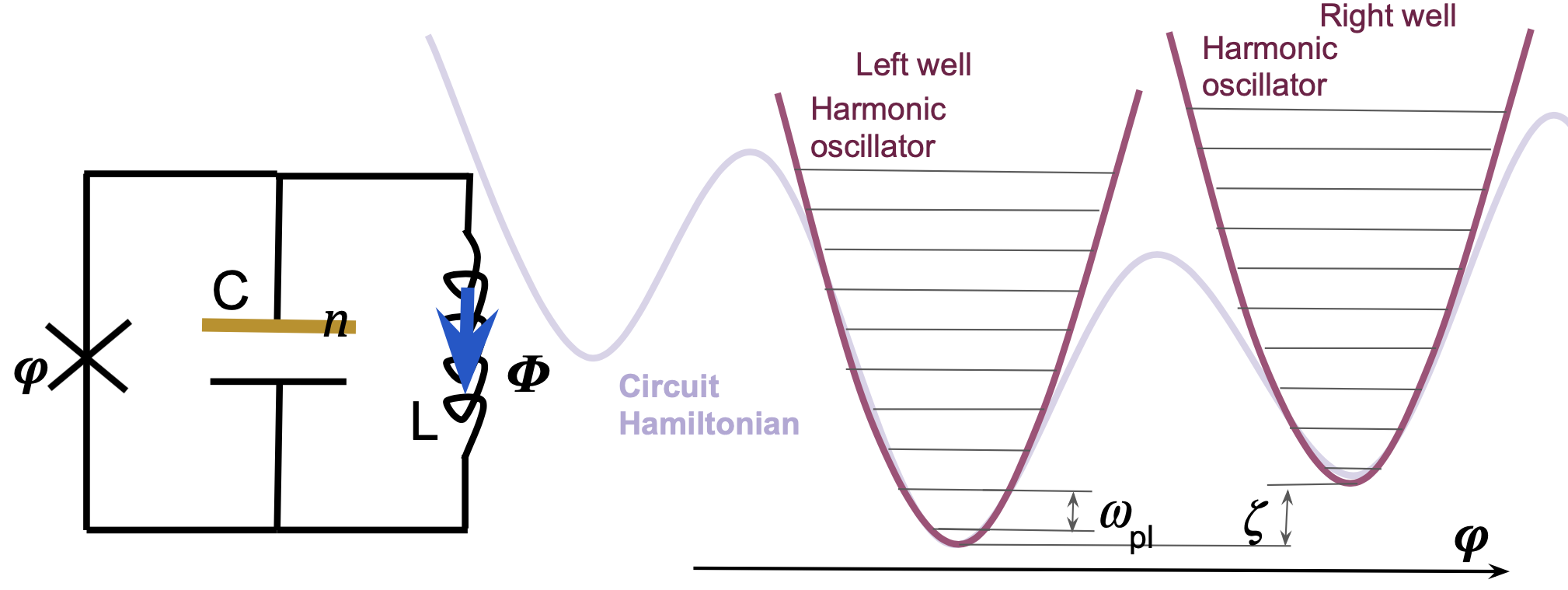}
\caption{The circuit corresponding to Eq. (\ref{eq:H-CJJ}), and the potential for the phase variable $\phi$. The lowest two wells are approximated as harmonic oscillators, with bias $\zeta$ and tunneling $\xi$ between the groundstates of the wells.}
\label{fig:dWone}
\end{figure}

Consider a D-Wave (CJJ rf SQUID) qubit~\cite{Harris:2008lp}. It consists of a large (main) loop and a small (CJJ) loop subjected to external flux biases $\Phi_x^q$ and $\Phi_x^{\text{CJJ}}$, respectively. The CJJ loop is interrupted by
two identical Josephson junctions connected in parallel with total capacitance $C$. For illustration purposes we represent this loop as a single junction with some external phase control in a circuit diagram in Figs.~\ref{fig:dWc} and~\ref{fig:dWone}. The two counter-circulating persistent current states along the main loop comprise the qubit $\ket{0}$ and $\ket{1}$ states, and can be understood as the states localized in the two wells of a double well potential, described below. 

The circuit Hamiltonian of this qubit can be written as in Eq.~\eqref{eq:H-CJJ}, where $n =Q/(2e)$ denotes the (normalized) quantized charge stored in the capacitance, $\phi = \frac{2\pi}{\Phi_0}\Phi$ is the (normalized) quantized total flux threading the main loop, $f=(2\pi/\Phi_0)\Phi_x^q$ and $E_Jb=-E_J^{\text{conventional}}\cos\left({\pi \Phi_x^{\text{CJJ}}}/{\Phi_0}\right)$ depend on the fluxes threading the main and small loops, respectively, $\Phi_0 = h/(2e)$ is the flux quantum (we use units of $h=1$ throughout), and $E_C= (2e)^2/(2C)$, $E_L=(\Phi_0/(2\pi))^2/(2L)$, and $E_J$ are the charging, (normalized) inductive, and Josephson energy, respectively. Note that the conventional notation for the Josephson energy translates to ours as $E_J^{\text{conventional}} =E_J \mathcal{B}$. The fluxes $\Phi_x^{\text{CJJ}}$ and $\Phi_x^q$ (and hence the parameters $b$ and $f$) are time-dependent and controllable, while the rest are fixed parameters set by the hardware. 

While $H_{\text{CJJ}}$ describes the physical circuit, we wish to implement the low energy Hamiltonian of a qubit with frequency $\omega_{\text{q}}$, as defined by Eq.~\eqref{eq:Hq}, using the approximate method given in Def.~\ref{imprec}. We now discuss how to make this transition. 
Treating the term $E_J b \cos {\phi} + E_L ({\phi}- f)^2$ as a classical potential in the variable $\phi$, it represents a cosine potential superimposed on a parabolic well. The two lowest states in this potential are the qubit states, separated by $\omega_{\text{q}}$. These two states need to be separated from non-qubit states, and the corresponding gap $\Delta$ is given by half the plasma frequency $\omega_{\text{pl}}$.

For a transmon, where $E_L=0$, one has $\omega_{\text{q}} = \omega_{\text{pl}}-E_C$~\cite{transmon-invention}, where the plasma frequency is given by 
\beq
\omega_{\text{pl}}(s) =2\sqrt{E_C E_J b(s)} .
\label{eq:omega_pl-transmon}
\eeq
Note that $b=1$ corresponds to when the cosine potential is shallowest, i.e., when the tunneling barrier is lowest, which is the initial point of the anneal with $s=0$. At the other extreme, when $b=\mc{B}$, the tunneling barrier is at its maximum and this corresponds to the end of the anneal with $s=1$.  

In the presence of the parabolic well there are additional levels in local minima of the raised cosine potential. For $f=0$ the two degenerate global minima appear at $\phi=\pm \pi$ and the lowest local minima at $\phi=\pm 3\pi$.  Thus, to ensure that the additional levels in the local minima are higher than the qubit frequency we can set $\min\omega_{\text{pl}}(s) = \omega_{\text{pl}}(0) \approx (\pm 3 \pi)^2E_L-(\pm\pi)^2E_L =8E_L \pi^2$.
Next, using $b(0)=1$, if $E_C \ll E_J$ (as it must, to ensure $\omega_{\text{q}}\ll \omega_{\text{pl}}$) then $E_L =O( \sqrt{E_C E_J}) \ll E_J$, which we will assume:
\beq
E_C,E_L \ll E_J .
\label{eq:116}
\eeq

We now wish to choose the controls of $H_{\text{CJJ}}$ so that $H_\text{q}$ in Definition \ref{imprec} takes the form:
\beq
H_{\text{q}}(s)=\xi(s) X + \zeta(s) Z ,
\label{eq:Hq2}
\eeq 
so that $\zeta(s) = \omega_{\text{q}} s$ [compare to Eq.~\eqref{eq:Hq}].
Focusing just on the minima at $\phi=\pm \pi$ but now allowing $f>0$, we have $ \zeta(s) = E_L(-\pi-f(s))^2 - E_L(\pi-f(s))^2 $, so that, upon neglecting the $f^2$ term:
\beq
   f(s) = \frac{\zeta(s)}{4E_L \pi} = s \frac{\omega_{\text{q}}}{4E_L \pi}  , \label{fDef}
\eeq
subject to $f(1)  <\pi$, i.e., we have the additional constraint $\omega_{\text{q}} <4 E_L \pi^2$.
 
Following Ref.~\cite{transmon-invention}, we can identify the bandwidth (peak-to-peak value for the charge dispersion of the energy levels in the periodic potential) of the $E_L =0$ Hamiltonian with the coefficient $\xi(s)$ in the effective qubit Hamiltonian. Under the assumed inequality~\eqref{eq:116}, Eq.~(2.5) of Ref.~\cite{transmon-invention} with $m=0$ yields
\begin{equation}
    \xi(s) = 8 E_C \sqrt{\frac{2}{\pi}} \left( \frac{2E_J b(s)}{E_C}\right)^{3/4} e^{-\sqrt{32b(s)E_J/E_C}} .
    \label{Koc}
\end{equation}
Thus, a sufficiently large $b(1)=\mathcal{B}$ ensures an exponentially small $\xi(1)$, which shows that we can operate the system in the annealing regime, i.e., the regime where $H_{\text{q}}(s)$ interpolates smoothly from $X$ to $Z$. Recall that $ b(0)=1$, thus $\xi(0) = \omega_{\text{q}}$ serves as a definition of $\omega_{\text{q}}$.
Let $\xi(1) /\omega_{\text{q}} = \delta$ be the desired precision. Then we can choose the remaining time dependent control $b(s)$ by solving Eq.~\eqref{Koc} for $b(s)$ and setting $\xi(s) = \omega_{\text{q}}(1-s + \delta)$ [again compare Eq.~\eqref{eq:Hq2} to Eq.~\eqref{eq:Hq}].
This together with Eq.~\eqref{fDef} fully defines the schedule. 

This mathematical model in fact describes a family of qubits, different by  $\omega_{\text{q}},\omega_{\text{pl}}(1)$ and $\delta$. The family is spanned by varying the ratio $E_J/E_C$  and $\mathcal{B}$,  in the region where both are $\gg1$ to ensure the applicability of Eq.~\eqref{Koc} and the smallness of the precision parameter $\delta$. Note that in the $E_J/E_C\gg1,\mathcal{B}\gg1$ regime the aforementioned conditions $\omega_{\text{q}} < \omega_{\text{pl}},4 E_L \pi^2$ are automatically satisfied. Among the qubits in the family, a smaller $\omega_{\text{q}}/\omega_{\text{pl}}(1)$ will allow a (relatively) faster anneal while the qubit approximation is maintained, but exactly how $E_J/E_C$ and $\mathcal{B}$ (or equivalently, $\omega_{\text{q}}/\omega_{\text{pl}}(1)$ and $\delta$) enter needs to be investigated via the adiabatic theorem, which we will delay until we analyze a simpler CSFQ case below.

We have thus shown how to reduce the circuit Hamiltonian $H_{\text{CJJ}}$ to an effective qubit Hamiltonian $H_{\text{q}}$, and how the circuit control functions $b(s)$ and $f(s)$ relate to the effective qubit annealing schedule functions $\xi(s)$ and $\zeta(s)$.

\subsubsection{Capacitively shunted flux qubit (CSFQ)}
\label{mainCalcSC}
We now repeat the analysis for a periodic $\phi$, i.e., for $H_{\text{CSFQ}}$ [Eq.~\eqref{eq:H-CSFQ}].
In this case the potential $E_J b \cos \phi - E_\alpha \cos\frac{1}{2}(\phi- f)$ exhibits only two wells.
For simplicity of the analysis, we instead choose to work with the Hamiltonian $H_{\text{CSFQ,sin}}$ given in Eq.~\eqref{eq:120}.
Recall that this Hamiltonian omits one of the terms in the trigonometric decomposition of $\cos\frac{1}{2}(\phi- f)$ and has the benefit that the wells are centered exactly at $\phi =\pm \pi$ for all $f$. Thus it ignores the diabatic effects from the wells shifting along the $\phi$ axis in the complete CSFQ Hamiltonian~\eqref{eq:H-CSFQ}. That effect can be included in the calculation straightforwardly, but for our example we choose the simplest nontrivial case.  Each well independently experiences narrowing as $b$ grows, leading to diabatic transitions out of the well's ground state. The physical meaning of the adiabatic timescale is to characterize the dynamics associated with this deformation of the harmonic oscillator, but by using the general machinery of our and the JRS bounds, we can obtain the result via algebra alone, without having to rely on physical intuition.

To apply the different versions of the adiabatic theorem expressed in Corollary~\ref{cor:JRScomp} we will need bounds on the derivatives of the simplified CSFQ Hamiltonian~\eqref{eq:120} (we drop the subscript and hat symbols for simplicity): 
\bes
\begin{align}
    H' &=E_J b' \cos \phi - \frac{E_\alpha}{2} f' \sin\frac{\phi}{2}\cos \frac{f}{2} \\
    H'' &=E_J b'' \cos \phi -\frac{E_\alpha}{2}  \sin\frac{\phi}{2} \left( f''\cos \frac{f}{2} - \frac{{f'}^2}{2}\sin \frac{f}{2}\right) .
\end{align}
\label{eq:H'H''}
\ees
In the JRS case one directly bounds the operator norm:
\bes
   \label{eq:121}
\begin{align}
   \label{eq:121a}
    \|H'\| &\leq E_J |b'| + \frac{E_\alpha}{2}|f'| \\
    \|H''\| &\leq E_J |b''|  + \frac{E_\alpha}{2}(|f''| +\frac{1}{2} |f'|^2) .
   \label{eq:121b}
\end{align}
\ees
In the case of our new version of the adiabatic theorem we will need bounds on the projected quantities. In any case, it is clear that we need to find bounds on the derivatives of $b$ and $f$, which we now proceed to derive.

\begin{figure}
\centering
\includegraphics[width=0.5\columnwidth]{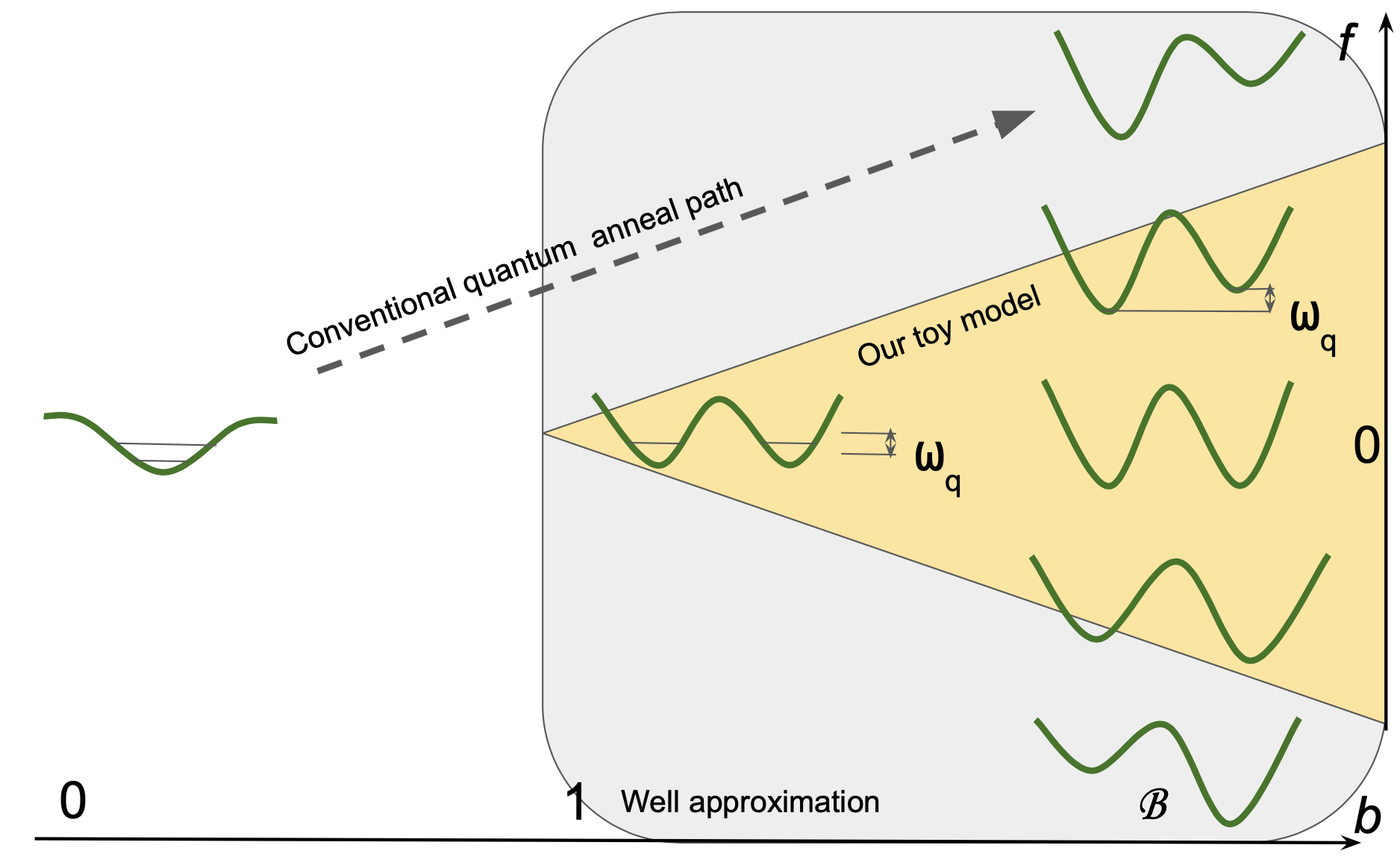}
\caption{The region in the space of control parameters $b(s), f(s)$ where quantum annealing of a flux qubit is analytically tractable within the well approximation.}
\label{fig:wellwell}
\end{figure}

\paragraph{The effective Hamiltonian.}
Define the well basis as the low energy basis diagonalizing $\phi$ projected into the low-energy subspace.
The qubit Hamiltonian in the well basis (see Definition \ref{imprec}) is:
\beq
H_{\text{q}}(s)=\xi(s) X + \zeta(s) Z ,
\eeq 
In the limit $E_\alpha \ll E_J$ we can approximate the width of the wells as equal, which leads to
\begin{equation}
    \zeta(s) \approx   
    E_\alpha \sin\frac{{\pi}}{2}\sin \frac{f}{2}-E_\alpha \sin\frac{{-\pi}}{2}\sin \frac{f}{2}
    = 2E_\alpha\sin\frac{1}{2} f(s) 
    \label{eq:122}
\end{equation}
[in this case the same result is obtained with the complete potential $E_\alpha \cos\frac{1}{2}(\phi- f)$].
We can also neglect the adjustment to the tunneling amplitude through the barrier of height $b E_J$ coming from the bias $ \zeta(s) \leq 2E_\alpha$ between wells. This again uses $E_\alpha \ll E_J$. Repeating the argument leading to Eq.~\eqref{Koc}, the zero-bias expression [Eq.~(2.5) of Ref.~\cite{transmon-invention} with $m=0$] holds for the tunneling amplitude, so we can reuse Eq.~\eqref{Koc}.
This expression also uses $E_C \ll E_J$. The more rigorous statement of the approximate equality in Eq.~\eqref{Koc} is postulated in the conjecture below. In Fig.~\ref{fig:wellwell} we contrast the special regime of these approximations, which we call the {\textit{well approximation}}, with the traditional schedule for quantum annealing.

\paragraph{Reducing the number of parameters.}
We choose the following notation for the ranges of $b$ and $f$:
\begin{equation}
    b:[0,1]\mapsto[1,\mathcal{B}], \quad f:[0,1]\mapsto[0,\mathcal{F}]
\end{equation}
In total, our CSFQ Hamiltonian has five parameters $E_C, E_J, E_\alpha, \mathcal{B}, \mathcal{F}$, i.e., four dimensionless parameters, since $\mathcal{B}$ and $\mathcal{F}$ are already dimensionless. We take $E_C$ to represent an overall energy scale, and define the dimensionless parameter $\mathcal{A}$ as the ratio appearing in $\xi(s)$, 
\beq
\mathcal{A} =\sqrt{\frac{32E_J }{E_C}}, 
\label{eq:def-A}
\eeq
rewriting Eq.~\eqref{Koc} as:
\begin{equation}
    \xi(s) \approx  E_C \sqrt{\frac{2}{\pi}} (\mathcal{A}\sqrt{b(s)})^{3/2} e^{-\mathcal{A}\sqrt{b(s)}}  . \label{xiForm}
\end{equation}
The parameter space can be reduced by setting $\mathcal{F} = \pi/3$. Note that the maximum allowed  $\mathcal{F}$ is $\pi$, at which $f'|_{s=1}$ required to fit the schedules will diverge. Making $\mathcal{F}$ really small just makes the qubit worse by adding additional constraints on other parameters, which justifies our choice. Then $f(1) = \pi/3$, so by Eq.~\eqref{eq:122} we have $E_\alpha =\zeta(1)$.

We now make use of $\omega_\text{q} =\xi(0) = \zeta(1) $. This means that the annealing schedule is such that the start and end energy approximately coincide, as is traditional for the idealized qubit model of annealing $ (1-s)X + sZ$. 
This allows us to write:
\beq
\omega_\text{q} =E_\alpha = \zeta(1) = \xi(0) = E_C \sqrt{\frac{2}{\pi}} \mathcal{A}^{3/2} e^{-\mathcal{A}} ,
\label{eq:126}
\eeq
i.e., the ratio $E_\alpha/E_C$ is also determined by $\mc{A}$. Having fixed the dimensionless parameters $E_J/E_C$ and $E_\alpha/E_C$ in terms of the single parameter $\mc{A}$, and having fixed $\mc{F}$ at a numerical value, we are left only with $\mc{A}$ and $\mc{B}$, i.e., we have reduced the original four dimensionless parameters to two. Let us now state the conjecture that replaces  Eq.~\eqref{Koc} by a rigorous statement:

\begin{myconjecture}
For a desired multiplicative precision $\epsilon$, there exists a minimum $\mathcal{A}_0(\epsilon)$ such that $\forall \mathcal{A}\geq \mathcal{A}_0$:
\begin{equation}
    \xi(s) =   E_C \sqrt{\frac{2}{\pi}} (\mathcal{A}\sqrt{b(s)})^{3/2} e^{-\mathcal{A}\sqrt{b(s)}}(1+ \epsilon'), \quad |\epsilon'| \leq \epsilon . \label{xiWeps}
\end{equation}
The two derivatives $\xi',\xi''$ are also given by the derivatives of Eq.~\eqref{xiForm} to the same multiplicative precision $\epsilon$.
\end{myconjecture}

The final transverse field needs to be negligible in quantum annealing. If our tolerance to a finite transverse field is $\delta$, then denote:
\begin{equation}
    \delta = \frac{\xi(1)}{\xi(0)} = (\mathcal{B}_0)^{3/4} e^{-\mathcal{A}(\sqrt{\mathcal{B}_0} -1) } .
    \label{eq:128}
\end{equation}
This implicitly defines $\mathcal{B}_0(\delta, \mathcal{A})>1$.
So our two dimensionless parameters live in a range $\mathcal{A} \in[\mathcal{A}_0(\epsilon), \infty]$ and  $\mathcal{B}\in[\mathcal{B}_0(\delta, \mathcal{A}),\infty]$.
Their physical meaning is: $\mathcal{A}$ is the (root of the) area under the barrier in appropriate dimensionless units at the beginning of the anneal, and $\mathcal{B}$ is how much the barrier has been raised at the end relative to the beginning. We note that both $\mathcal{B}_0, \mathcal{A}_0$ are rather large numbers for reasonable $\epsilon$ and $\delta$.\footnote{For $\epsilon=10^{-1}, \delta=10^{-9}$ we are free to choose $\mathcal{A}_0$ satisfying Eq.~\eqref{xiWeps}. For $b=1$, if we assume $\epsilon\sim \frac{1}{\mathcal{A}_0}$ as well as subleading exponential terms, this would lead to an estimate $\mathcal{A}_0=10$. Now solving Eq.~\eqref{eq:128} for $\mathcal{B}_0$, we find $\mathcal{B}_0\approx 10.6$}, thus we intend to investigate the scaling of the adiabatic timescale $\theta$ in the limit $\mathcal{A}, \mathcal{B} \to \infty$. The relationship between $\mathcal{A}$ and $\mathcal{B}$ as they approach that limit may be arbitrary; we do not make any additional assumptions about this. 

The gap to the non-qubit states is, to the leading order, determined by the plasma frequency:
\begin{equation}
    \omega_{\text{pl}} (b) = 2\sqrt{E_C E_J b} 
    = E_C \mathcal{A}\sqrt{b(s)/8} ,
    \label{eq:omegapl}
\end{equation}
which is the same as Eq.~\eqref{eq:omega_pl-transmon} for the D-Wave qubit.  Even though $\omega_{\text{pl}} (b)$ attains its minimum value at $b(1)=1$, we will find that the terms in the numerator of the adiabatic theorem overwhelm it in such a way that only $ \omega_{\text{pl}}(\mathcal{B})$ at the end of the anneal matters.

Repeating the reasoning of the CJJ qubit case above, $\xi(0)=\omega_{\text{q}}$ serves as the definition of $\omega_{\text{q}}$, 
and the time dependent controls $f(s), b(s)$ should be [approximately, using Eq.~\eqref{xiForm}] chosen as:
\bes
\label{schedCSFQ} 
\begin{align}
   \label{eq:130a}
           \frac{\zeta(s)}{\zeta(1)} =2\sin\frac{1}{2} f(s) &=s , \\
   \frac{\xi(s)}{\xi(0)} = b(s)^{3/4} e^{-\mathcal{A}\sqrt{b(s)} +\mathcal{A}} &= 1-s + \delta_\mathcal{B} . 
   \label{eq:130b}
\end{align}
\ees
Here $\delta_\mathcal{B} \le \delta$\footnote{This holds since $\delta = \delta_{\mathcal{B}_0}$ and $\mathcal{B}>\mathcal{B}_0>1$, and the function $\delta_{\mathcal{B}}$ is monotonically decreasing in $\mathcal{B}$ for $\mathcal{B}>\Theta(1/\mathcal{A}^2)$.} is the precision we get for this choice of $\mathcal{B}$. The quantity $\delta_\mathcal{B}$ together with the ratio of the qubit frequency $\omega_{\text{q}} = \xi(0) = E_{\alpha}$ [Eq.~\eqref{eq:126}] to the plasma frequency at the end of the anneal $\omega_{\text{pl}}(\mathcal{B})  = E_C \mathcal{A}\sqrt{\mathcal{B}/8}$, are the two independent parameters we will use to present the final answer for $ \theta^{\text{new}}$. The relationship of these two parameters with $\mathcal{A},\mathcal{B}$ is given by:
\begin{equation}
    \delta_\mathcal{B} =  \mathcal{B}^{3/4} e^{-\mathcal{A}(\sqrt{\mathcal{B}} -1) }, \quad \frac{\omega_\text{q}}{\omega_{\text{pl}}(\mc{B})} =   \frac{4}{\sqrt{\pi}} \sqrt{\frac{\mathcal{A}}{\mathcal{B}}} e^{-\mathcal{A}} .
    \label{eq:132}
\end{equation}

\paragraph{The derivatives $b', b'', f',f''$.} 
First, from Eq.~\eqref{eq:130a}:
\bes
\begin{align}
\label{eq:f'}
         f'(s) &=\frac{1}{\cos (f/2)} = \frac{1}{\sqrt{1-(s/2)^2}} \leq \frac{2}{\sqrt{3}} \\
 f''(s) &= \frac{s/4}{(1-(s/2)^2)^{3/2}} \leq \frac{2}{3\sqrt{3}} .
         \label{eq:f''}
         \end{align}
\ees
Second, from Eq.~\eqref{eq:130b}: 
\begin{align}
   b' \left(\frac{3}{4b} - \frac{\mathcal{A}}{2\sqrt{b}} \right)b(s)^{3/4} e^{-\mathcal{A}\sqrt{b(s)} +\mathcal{A}}&= -1  \quad \Rightarrow\quad b' \left(\frac{3}{4\sqrt{b}} - \frac{\mathcal{A}}{2} \right) = -b(s)^{-1/4} e^{\mathcal{A}\sqrt{b(s)} - \mathcal{A}} .
\end{align}
Since $\mathcal{A}\gg1, b\geq1$, we can neglect the subleading term $\frac{3}{4\sqrt{b}}$, i.e.,
\beq
b'\approx \frac{2}{\mc{A}b^{1/4}(s)} e^{\mathcal{A}\sqrt{b(s)} - \mathcal{A}} .
\label{eq:b'}
\eeq 
We do the same in the calculation of the second derivative: 
\begin{equation}
\label{derivs}
b'' \approx \frac{b'}{b(s)^{3/4}}  e^{\mathcal{A}(\sqrt{b(s)} -1)} 
    \approx \frac{2}{\mathcal{A}b(s)}  e^{2\mathcal{A}(\sqrt{b(s)} -1)} .
\end{equation}
We will use a change of integration variable
\begin{equation}
\label{changeOfVar}
    ds =\frac{\mathcal{A}b^{1/4}(s)}{2}  e^{-\mathcal{A}(\sqrt{b(s)} -1)} db .
\end{equation}
We also note that $b', b''$ are exponentially large in $\mathcal{A}(\sqrt{b(s)}-1)$, thus they have the potential of becoming the leading terms in our estimate for the adiabatic timescale.  

\paragraph{Completing the proof of the result claimed in Eq.~\eqref{eq:113}.} 
We show below that $\|H'\|$ does not grow with the cutoff $\Lambda$, so we apply Corollary~\ref{cor:JRScomp}. Using the JRS formula~\eqref{eq:theta-JRS} with $d=2$ and $\D \approx \omega_{\text{pl}}/2$, we have:
\begin{align}
    \theta^{\text{JRS}}(s^*) \approx\frac{8\|H'(0)\|}{\omega^2_{\text{pl}}(b(0))} +  \frac{8\|H'(s^*)\|}{\omega^2_{\text{pl}}(b(s^*))} + I  \ , \qquad I \equiv \int_0^{s^*}  \left( \frac{8 \|H''(s)\|}{\omega^2_{\text{pl}} (b(s))} + 7\cdot2^4\sqrt{2} \frac{\|H'(s)\|^2}{\omega^3_{\text{pl}} (b(s))}\right) ds .
    \label{eq:I}
\end{align}

Returning to Eq.~\eqref{eq:121a}, we now substitute the derivatives of $b$ and $f$ we found in terms of $\mathcal{A},b$, using Eqs.~\eqref{eq:def-A},~\eqref{eq:126}, \eqref{eq:f'}, and~\eqref{eq:b'}:
\begin{equation}
    \|H'(s)\| \leq  \frac{E_C}{32}\mathcal{A}^2 \left(\frac{2}{\mathcal{A}b^{1/4}(s)}  e^{\mathcal{A}(\sqrt{b(s)} -1)}\right) (1+ o(1)) + E_C \sqrt{\frac{2}{\pi}} \mathcal{A}^{3/2} e^{-\mathcal{A}}\frac{2}{\sqrt{3}} ,
    \label{eq:139}
\end{equation}
where the $o(1)$ accounts for the term we neglected in approximating $b'$ to arrive at Eq.~\eqref{eq:b'}.
The second term in Eq.~\eqref{eq:139} (arising from $f'$) is subleading, and since we only kept the leading term in the derivatives of $b$, we should omit it. The same happens for the second derivative, for which we use Eqs.~\eqref{eq:121b} and~\eqref{derivs}. Thus:
\bes
\label{Hprwins}
\begin{align}
\label{Hprwins-a}
     \|H'(s)\| &\leq  \frac{E_C}{32}\mathcal{A} \frac{2}{b^{1/4}(s)}  e^{\mathcal{A}(\sqrt{b(s)} -1)} (1+ o(1))\\
      \|H''(s)\| &\leq \frac{E_C}{32}\mathcal{A}\frac{2}{b(s)}  e^{2\mathcal{A}(\sqrt{b(s)} -1)}(1+o(1)) .
      \label{Hprwins-b}
\end{align}
\ees
Here $o(1)$ means going to zero in the limit $\mathcal{A} \to \infty$, or $b\to \infty$. We will omit the $(1+o(1))$ clause below when working with leading order expressions.

Let us substitute the expressions obtained so far into the integral $I$ [Eq.~\eqref{eq:I}], and change variables to $db$ using Eq.~\eqref{changeOfVar}:
\begin{equation}
     I \le \int_1^{b(s^*)}  \left( \frac{8 \frac{E_C}{32}\mathcal{A}\frac{2}{b}  e^{2\mathcal{A}(\sqrt{b} -1)}}{E_C^2 \mathcal{A}^2b/8} + 7\cdot2^4\sqrt{2} \frac{ \frac{E_C^2}{32^2}\mathcal{A}^2 \frac{4}{\sqrt{b}}  e^{2\mathcal{A}(\sqrt{b} -1)}}{ E_C^3 \mathcal{A}^3(b/8)^{3/2}}\right) \frac{\mathcal{A}b^{1/4}}{2}  e^{-\mathcal{A}(\sqrt{b} -1)} db ,
\end{equation}
where we also used Eq.~\eqref{eq:omegapl}.
The two terms depend on $\mathcal{A}$ and $b$ in exactly the same way: 
\begin{equation}
    E_CI \le 9\int_1^{b(s^*)}b^{-7/4}e^{\mathcal{A}(\sqrt{b}-1)}db = 18\mathcal{A}^{3/2}e^{-\mathcal{A}}\int_\mathcal{A}^{\mathcal{A}\sqrt{b(s^*)}}w^{-5/2}e^{w}dw .
    \label{eq:142}
\end{equation}
The integral can be computed analytically in terms of the exponential integral function, but it is more insightful to observe that it is dominated by the upper integration limit, under the assumption that $b(s^*) \gg 1$. Indeed since $B\gg 1$, there is a range of $s^*$ close to $1$ for which Eq. (\ref{eq:130b}) gives $b(s^*)\gg 1$. In that regime:
\begin{equation}
    \int_\mathcal{A}^{\mathcal{A}\sqrt{b(s^*)}}{w^{-5/2}}e^{w} dw \approx \int_{-\infty}^{\mathcal{A}\sqrt{b(s^*)}}\frac{1}{(\mathcal{A}\sqrt{b(s^*)})^{5/2}}e^{w}dw  = \frac{e^{\mathcal{A}\sqrt{b(s^*)}}}{(\mathcal{A}\sqrt{b(s^*)})^{5/2}} .
    \label{eq:143}
\end{equation}
Hence:
\begin{equation}
     I \leq 18  \frac{e^{\mathcal{A}(\sqrt{b(s^*)}-1)}}{E_C\mathcal{A}b(s^*)^{5/4}} .
    \label{eq:144}
\end{equation}
The full bound for $\theta$ is thus, using Eqs.~\eqref{eq:omegapl},~\eqref{eq:I},~\eqref{Hprwins}, and~\eqref{eq:144}:
\begin{align}
    \theta^{\text{JRS}}(s^*) 
    \leq \frac{4}{E_C \mathcal{A}}\left(1+
    \frac{ e^{\mathcal{A}(\sqrt{b(s^*)} -1)}}{b(s^*)^{5/4}}+ 
    \frac{9}{2} \frac{e^{\mathcal{A}(\sqrt{b(s^*)}-1)}}{b(s^*)^{5/4}}  \right) \approx\\\approx 22  \frac{e^{\mathcal{A}(\sqrt{b(s^*)}-1)}}{E_C\mathcal{A}b(s^*)^{5/4}} = 22  \frac{1}{(1-s^* +\delta_{\mathcal{B}})E_C\mathcal{A}\sqrt{b(s^*)}} ,
    \label{eq:145}
\end{align}
where neglecting the subleading first term (arising from $s=0$) means that only the end of the anneal matters, and we used Eq.~\eqref{eq:132} to obtain the last equality.
Reintroducing $\omega_{\text{pl}} (b(s^*))  = E_C \mathcal{A}\sqrt{b(s^*)/8}$ [Eq.~\eqref{eq:omegapl}], we obtain:
\begin{equation}
    \omega_{\text{q}}\theta^{\text{JRS}}(s^*) \leq  \frac{11}{\sqrt{2}}\frac{1}{(1-s^* +\delta_{\mathcal{B}})}\frac{\omega_{\text{q}}}{\omega_{\text{pl}} (b(s^*))}  .
    \label{eq:JRS-CSFQ}
\end{equation}
The ratio of qubit frequency over gap is what one would intuitively expect from the adiabatic theorem, but the other factors can only be obtained after a detailed calculation such as the one performed here.

\paragraph{Completing the proof of the result claimed in Eq.~\eqref{eq:114}.} 
Since we have already shown that $H'$ does not grow with the cutoff $\Lambda$ [Eq.~\eqref{Hprwins-a}] we now use Eq.~\eqref{eq:theta-new} [Corollary~\ref{cor:JRScomp}] for the CSFQ Hamiltonian. 

It turns out that there is no benefit from the projection in $\|PH'P\|$ so we just use $\|PH'P\|\leq \|H'\|$, and focus on the off-diagonal terms $\|PH''Q\|$ and $\|PH'Q\|$ to obtain an improvement over the JRS bound~\eqref{eq:JRS-CSFQ}.  Starting from Eq.~\eqref{eq:H'H''}, we have:
\bes
\begin{align}
    PH'Q &=E_J b' P\cos \hat{\phi} Q - \frac{E_\alpha}{2} f' P\sin\frac{\hat{\phi}}{2} Q\cos \frac{f}{2} \\
    PH''Q &=E_J b'' P\cos\hat{\phi} Q -\frac{E_\alpha}{2}  P\sin\frac{\hat{\phi}}{2}Q \left( f''\cos \frac{f}{2} - \frac{{f'}^2}{2}\sin \frac{f}{2}\right) .
\end{align}
\ees
Thus we need to estimate the leading order of the bound on $\|P\cos \hat{\phi} Q\|$ and $\|P\sin\frac{\hat{\phi}}{2}Q\|$.
For this estimate, we make use of the well approximation: the eigenstates are approximately the states of a harmonic oscillator centered at each well (Fig.~\ref{fig:wellwell}). Indeed, recall that $H_{\text{CSFQ,sin}}$ [Eq.~\eqref{eq:120}] is a Hamiltonian representing a double-well potential centered exactly at $\phi =\pm \pi$ for all $f$. We thus approximate $H_{\text{CSFQ,sin}}$ as the sum of 
\begin{equation}
    H_{L} = E_C{\hat{n}}^2 + \frac{1}{2}E_Jb (\hat{\phi} + \pi)^2\ , \quad H_{R} = E_C{\hat{n}}^2 +\frac{1}{2} E_Jb (\hat{\phi} - \pi)^2 .
\end{equation}
$P$ projects onto the span of the ground states of these two Hamiltonians, while $Q$ projects onto the span of the first and higher excited states. Denote $\delta \hat{\phi}_{L,R}\equiv \hat{\phi}\pm \pi$, then the expression for the position operators $\delta\phi$ in terms of the corresponding harmonic oscillator creation and annihilation operators is:\footnote{To see this, consider the standard 1D quantum harmonic oscillator Hamiltonian $H = \alpha {\hat{p}^2}+\beta\hat{x}^2$, where $\alpha = \frac{1}{2m}$ and $\beta = \frac{1}{2}m\omega^2$, which after the introduction of the standard creation and annihilation operators gives $\hat{x} = \gamma(\hat{a}+\hat{a}^\dagger)$, where $\gamma = ({\frac{\hbar^2\alpha}{4\beta}})^{1/4} = \sqrt{\frac{\hbar}{2m\omega}}$; in our case $\alpha = E_C $ and $\beta =  \frac{1}{2}E_Jb$, so that $\gamma =\Theta\left(\frac{E_C}{E_Jb} \right)^{1/4}$.}
\begin{equation}
    \delta \hat{\phi}_{L,R} = O\left(\frac{E_C}{E_Jb} \right)^{1/4} (a_{L,R} + a_{L,R}^\dag) . 
    \label{coordOp}
\end{equation}
We can now estimate:
\begin{equation}
      \|P\cos \hat{\phi} Q\| \approx \||g_L\rangle\!\langle g_L| \cos \delta\hat{\phi}_L (1-|g_L\rangle\!\langle g_L|) +|g_R\rangle\!\langle g_R| \cos \delta\hat{\phi}_R (1-|g_R\rangle\!\langle g_R|) \| ,
\end{equation}
where $|g_{L,R}\rangle$ are the ground states in the corresponding wells, and we neglected the matrix elements of $\cos\phi$ that mix the wells. We proceed as follows:
\bes
\begin{align}
     \|P\cos \hat{\phi} Q\| &\approx
     \||g_L\rangle\!\langle g_L| (1-\frac{1}{2} \delta\hat{\phi}^2_L) (1-|g_L\rangle\!\langle g_L|) +|g_R\rangle\!\langle g_R| (1-\frac{1}{2} \delta\hat{\phi}^2_R) (1-|g_R\rangle\!\langle g_R|) \| \\ 
     &=
     \frac{1}{2}\||g_L\rangle\!\langle g_L|  \delta\hat{\phi}^2_L (1-|g_L\rangle\!\langle g_L|) +|g_R\rangle\!\langle g_R|  \delta\hat{\phi}^2_R (1-|g_R\rangle\!\langle g_R|) \|  \\ 
     &\leq \frac{1}{2}(\||g_L\rangle\!\langle g_L| \delta\hat{\phi}^2_L (1-|g_L\rangle\!\langle g_L|)\| +\||g_R\rangle\!\langle g_R| \delta\hat{\phi}^2_R (1-|g_R\rangle\!\langle g_R|) \|) .
\end{align}
\ees
Plugging in Eq.~\eqref{coordOp},\footnote{In the number basis we have $|g\rangle\!\langle g| (a + a^\dag)^2 (1-|g\rangle\!\langle g|) = \ket{0}\bra{0}(a + a^\dag)^2\sum_{n=1}^\infty\ket{n}\!\bra{n} = \sqrt{2}\ket{0}\!\bra{2}$, and $\|\ket{0}\!\bra{2}\| = 1$ (largest eigenvalue of $\ket{2}\!\bra{2}$).}
and repeating the same calculation for $\|P\cos\frac{\delta\hat{\phi}}{2}Q\|$, we get:
\bes
\begin{align}
    & \|P\cos \hat{\phi} Q\| = O\left(\frac{E_C}{E_Jb} \right)^{1/2} \\
    & \|P\sin\frac{\hat{\phi}}{2}Q\| \approx \|P\cos\frac{\delta\hat{\phi}}{2}Q\| =  O\left(\frac{E_C}{E_Jb} \right)^{1/2}
\end{align}
\ees
Thus the expressions~\eqref{eq:121} and~\eqref{Hprwins} get multiplied by the same factor $O\left(\frac{E_C}{E_Jb} \right)^{1/2}$:
\bes
\begin{align}
\label{eq:151a}
\|PH'(s)Q\| &\approx \|H'(s)\| O\left(\frac{E_C}{E_Jb(s)} \right)^{1/2}\\
\|PH''(s)Q\| &\approx \|H''(s)\| O\left(\frac{E_C}{E_Jb(s)} \right)^{1/2} .
\end{align}
\ees
Note that since $\sqrt{E_C/E_J} = \sqrt{32}/\mathcal{A}$ [Eq.~\eqref{eq:def-A}] and $b(s)\in [1,\mc{B}]$, we have 
\beq
\left(\frac{E_C}{E_Jb(s)} \right)^{1/2}\ll 1 \ \ \forall s .
\label{eq:152}
\eeq
We can carry the $O\left(\frac{E_C}{E_Jb} \right)^{1/2}$ factor through the calculations all the way until the integration, as in Eq.~\eqref{eq:I}, except that now the integral is the one appearing in Eq.~\eqref{eq:theta-new}. Thus, again using $d=2$ and $\D \approx \omega_{\text{pl}}/2$, and absorbing all numerical factors into $O(1)$ when convenient: 
\bes
\begin{align}
I &\approx  \int_0^{s^*}\left(\frac{8\|PH''(s)Q\|}{\omega^2_{\text{pl}}(b(s))} + 2^4\sqrt{2}\frac{\|PH'(s)Q\|(5\|PH'(s)Q\| + 3\|PH'(s)P\| + 3\|H'(s)\|)}{\omega^3_{\text{pl}} (b(s))} \right)ds \\
      &\le O(1)\int_0^{s^*}  \left( \frac{ \|PH''(s)Q\|}{\omega^2_{\text{pl}} (b(s))} +  \frac{\|PH'(s)Q\|(\|PH'(s)Q\|+\|H'(s)\|)}{\omega^3_{\text{pl}} (b(s))}\right) ds .
\end{align}
\ees
It follows from Eqs.~\eqref{eq:151a} and~\eqref{eq:152} that we may neglect $\|PH'(s)Q\|$ compared to $\|H'(s)\|$.
We may thus proceed from Eq.~\eqref{eq:142} but multiply the right-hand side by $O\left(\frac{E_C}{E_Jb(s)} \right)^{1/2} = O(1)\frac{1}{\mc{A}\sqrt{b(s)}}$:
\begin{align}
    E_CI \le O(1)\left(\frac{E_C}{E_J} \right)^{1/2} \int_1^{b(s^*)}b^{-9/4}e^{\mathcal{A}(\sqrt{b}-1)}db =\\ = O(1)\mathcal{A}^{3/2}e^{-\mc{A}}\int_\mathcal{A}^{\mathcal{A}\sqrt{b(s^*)}}w^{-7/2}e^{w} dw \approx O(1) \frac{e^{\mathcal{A}(\sqrt{b(s^*)}-1)}}{\mathcal{A}^2{b(s^*)}^{7/4}} ,
    \label{eq:154}
\end{align}
where in the last approximate equality we applied the same reasoning as in Eq.~\eqref{eq:143}. 

Comparing to the latter, we see that  the expression gained an overall factor of $\frac{1}{\mathcal{A}\sqrt{b(s^*)}}$. The same happens with the leading boundary term. Using Eq.~\eqref{Hprwins-a}:
\bes
\begin{align}
& \left.\frac{d\|PH'Q\|}{\D^2}\right|_{s=0} +\left.\frac{d\|PH'Q\|}{\D^2}\right|_{s=s^*} \lesssim O(1)\left[\left(\frac{E_C}{E_Jb(0)} \right)^{1/2} \frac{\|H'(0)\|}{\omega^2_{\text{pl}} (b(0))}+ \left(\frac{E_C}{E_Jb(1)} \right)^{1/2} \frac{\|H'(s^*)\|}{\omega^2_{\text{pl}} (b(s^*))} \right] \\
 &\qquad = O(1)  \frac{1}{E_C\mathcal{A}^2} \left(1+\frac{e^{\mc{A}(\sqrt{b(s^*)}-1)}}{{b(s^*)^{7/4}}}\right) \approx O(1)  \frac{e^{\mc{A}(\sqrt{b(s^*)}-1)}}{{E_C\mathcal{A}^2b(s^*)^{7/4}}} ,
\end{align}
\ees
which is of the same order as the integral term.
Thus:
\bes
\label{compare1}
\begin{align}
\theta^{\text{new}}(s^*) &= \left.\frac{d\|PH'Q\|}{\D^2}\right|_{s=0} +\left.\frac{d\|PH'Q\|}{\D^2}\right|_{s=s^*} +I  
\lesssim O(1)  \frac{e^{\mathcal{A}(\sqrt{b(s^*)}-1)}}{E_C\mathcal{A}^2b(s^*)^{7/4}} \\
\label{compare1-b}
&= O(1) \frac{1}{\mathcal{A}\sqrt{b(s^*)}} \theta^{\text{JRS}}(s^*) \\
&= O(1) \frac{1}{(1-s^*+\delta_{\mathcal{B}})E_C\mathcal{A}^2{b(s^*)}}  = O(1)\frac{1}{(1-s^* +\delta_{\mathcal{B}})}
\frac{E_C}{\omega_{\text{pl}}^2 (b(s^*))}   ,
\label{compare1-c}
\end{align}
\ees
where in the second line we used Eq.~\eqref{eq:145} and $\omega_{\text{pl}} (b(s^*))  = \Theta(1) E_C \mathcal{A}\sqrt{b(s^*)}$ [Eq.~\eqref{eq:omegapl}].

Now, using $\omega_{\text{pl}} (b(0)) = E_C \mathcal{A}\sqrt{1/8}$ we have $E_C = \Theta(1) \omega_{\text{pl}} (b(0))/\mc{A}$. Also, from Eqs.~\eqref{eq:126} and~\eqref{eq:omegapl} we have $\mathcal{A}^{1/2}e^{-\mathcal{A}} = \Theta(\omega_{\text{q}}/\omega_{\text{pl}}(b(0)))$, which we can solve  approximately as $\mathcal{A} = \Theta(\text{ln}(\omega_{\text{pl}}(b(0))/\omega_{\text{q}}))$. Combining this with Eq.~\eqref{compare1-c}, we get:
\bes
\begin{align}
    \theta^{\text{new}}(s^*) &\lesssim O(1)\frac{1}{(1-s^* +\delta_{\mathcal{B}})}\frac{E_C}{\omega_{\text{pl}}^2 (b(s^*))}    = 
    O(1)\frac{1}{(1-s^* +\delta_{\mathcal{B}})}\frac{\omega_{\text{pl}} (b(0))}{\omega^2_{\text{pl}} (b(s^*))\text{ln}(\omega_{\text{pl}}(b(0))/\omega_{\text{q}})} \\
    &= O(1) \theta^{\text{JRS}}(s^*) \frac{\omega_{\text{pl}} (b(0))}{\omega_{\text{pl}} (b(s^*))\text{ln}(\omega_{\text{pl}}(b(0))/\omega_{\text{q}})} ,
\end{align}
\ees
where the JRS result is given in Eq.~\eqref{eq:JRS-CSFQ}.

\subsubsection{Comparison of the two bounds for the CSFQ}

To compare the two bounds, it is useful to express everything via two parameters at $s^*$ only: $1-s^* +\delta_{\mathcal{B}}$ and $\frac{\omega_{\text{q}}}{\omega_{\text{pl}} (b(s^*))} $. Note, combining Eqs.~\eqref{eq:126},~\eqref{eq:omegapl} and~\eqref{eq:130b}, that:
\begin{align}
  (1-s^*+ \delta_{\mathcal{B}}) \frac{\omega_{\text{q}}}{\omega_{\text{pl}} (b(s^*))} = \frac{4}{\sqrt{\pi}}  (\mathcal{A}\sqrt{b(s^*)})^{1/2} e^{-\mathcal{A}\sqrt{b(s^*)}} \quad \Rightarrow \\ \Rightarrow \quad \mathcal{A}\sqrt{b(s^*)} = -(1+ o(1))\ln\left[(1-s^* +\delta_{\mathcal{B}}) \frac{\omega_{\text{q}}}{\omega_{\text{pl}} (b(s^*))} .\right]
\end{align}
Thus, since Eq.~\eqref{compare1-b} shows that the new bound is related to the JRS bound by the factor $1/(\mc{A}\sqrt{b(s^*)})$, 
using the new bound leads to a logarithmic correction of the original adiabatic timescale:
\begin{equation}
    \theta^{\text{new}}(s^*) =\theta^{\text{JRS}}(s^*)\frac{ O(1)}{-\ln\left[(1-s^* +\delta_{\mathcal{B}})\frac{\omega_{\text{q}}}{\omega_{\text{pl}} (b(s^*))}\right]} .
\label{eq:159}
\end{equation}
We conclude that there are two competing small numbers, 
$1-s^* +\delta_{\mathcal{B}}$ and $\frac{\omega_{\text{q}}}{\omega_{\text{pl}} (b(s^*))}$. The gap to the third state should be much larger than the qubit frequency, i.e., $\omega_{\text{pl}}(b(s))\gg\omega_{\text{q}}$ $\forall s$. The expression $1-s^* +\delta_{\mathcal{B}}$ [recall its definition in Eq.~\eqref{eq:130b}], times $\omega_{\text{q}}$, can be interpreted as a residual transverse field $h_x$ at $s=s^*$.
This residual transverse field should satisfy $h_x /\omega_{\text{q}} = 1-s^* +\delta_{\mathcal{B}}\ll 1$ in the regime where the expression $\theta(s^*)$ for the adiabatic timescale over the interval $[0,s^*]$ is valid. Using Eqs.~\eqref{eq:JRS-CSFQ} and~\eqref{eq:159} we may rewrite the two bounds as:
\begin{equation}
\label{thetasFinal}
    \omega_{\text{q}}\theta^{\text{JRS}}(s^*) = O(1)\left.\frac{\omega_{\text{q}}^2}{\omega_{\text{pl}} h_x} \right|_{s=s^*}, \quad  \omega_{\text{q}}\theta^{\text{new}}(s^*) = O(1)\left.\frac{\omega_{\text{q}}^2}{\omega_{\text{pl}} h_x \text{ln}\frac{\omega_{\text{pl}} }{h_x}} \right|_{s=s^*} .
\end{equation}
Thus, if the geometric mean $\sqrt{h_x\omega_{\text{pl}}} \gg \omega_{\text{q}}$, then the effective dynamics stays within the qubit approximation well. Our new bound adds a logarithmic correction to this estimate, and is tighter than the JRS bound since $\omega_{\text{pl}}(b(s^*)) > h_x$. Finally, we note that a brute-force calculation we present in Appendix~\ref{app:A} obtains an equivalent bound. 

Since the adiabatic timescale increases as $s^*$ approaches $1$, there is a regime of intermediate anneal times $t_f$ such that:
\begin{equation}
    \theta(0)\leq  t_f \leq \theta(1) \quad \Rightarrow \quad \frac{11/\sqrt{2}}{\omega_{\text{pl}}(0)} \leq t_f \leq \frac{11/\sqrt{2}}{\delta_{\mc{B}}\omega_{\text{pl}}(1)} \ ,
\end{equation}
where we dropped the logarithmic corrections, and also for the purposes of estimation used $\theta(0)$ even though this is outside the range of applicability of our expression for $\theta$. In this regime there is $s^*$ such that $t_f =\theta(s^*)$, and the physical intuition is that the anneal over the interval $[0,s^*]$ stays within the qubit approximation, while the anneal beyond that in the interval $[s^*,1]$ leaves the qubit subspace. We do not know if there is still an effective qubit description of this dynamics, but we note that it is not likely to be given by the dynamics of the lowest levels alone. Indeed, although there will still be tunneling between the wells in $[s^*,1]$, there is no clear way to define a phase of the state in each well, since that state involves several energy levels of that well. Thus the pattern of interference that emerges when the populations of two wells meet after tunneling will no longer be governed by a single phase parameter. This intuition suggests that either a full multilevel description should be utilized instead of a qubit description, or possibly there is an effective stochastic description that arises after we neglect any interference effects but keep the dimension of the qubit model. The development of such a theory is beyond the scope of this work.

\subsubsection{Bound for CJJ}
To obtain a meaningful expression for the adiabatic timescale $\theta$ for the CSFQ qubit case above, we had to use a ``well approximation": the two wells of the $\phi$-potential of the Hamiltonian
\begin{equation}
     H_{\text{CSFQ,sin}} = E_C \hat{n}^2 + E_J b \cos \hat{\phi} - E_\alpha \sin\frac{\hat{\phi}}{2}\sin \frac{f}{2}  \quad \phi \in [-2\pi, 2\pi] 
\end{equation}
are separated by a large enough barrier $\sim bE_J$ throughout the anneal, so that the low energy subspace is approximately given by the ground states of the harmonic approximation of the left and right wells:
\bes
\label{HOequiv}
\begin{align}
    H_{\text{CSFQ,L}} = E_C \hat{n}^2 + E_J b \frac{ (\hat{\phi}+\pi)^2}{2},  \quad \phi \in [-\infty, \infty], \\
    H_{\text{CSFQ,R}} = E_C \hat{n}^2 + E_J b \frac{ (\hat{\phi}-\pi)^2}{2},   \quad \phi \in [-\infty, \infty] .
\end{align}
\ees
Note that we neglected the adjustment of the harmonic potential by the last term, and made a constant energy shift $\pm E_\alpha \sin \frac{f}{2}$. 
If we choose $b(s)$ and $f(s)$ in the same way as in Eq.~\eqref{schedCSFQ}, the derivatives $H_{\text{CSFQ,L}}'$ and $H_{\text{CSFQ,R}}'$ become arbitrarily large with the cutoff $\Lambda$, so the JRS bound will no longer be cutoff-independent. We will need to use the relation
\begin{equation}
   (H_{\text{CSFQ,L}}')^2 \leq  \frac{{b'}^2}{b^2}H_{\text{CSFQ,L}}^2 .
\end{equation}
Applying our adiabatic theorem [Eq.~\eqref{adtime}] to staying in the ground state of $H_{\text{CSFQ,L}}$ and $H_{\text{CSFQ,R}}$, we will find that a term with this extra factor ${b'}/b$ turns out to be subleading. We do not present the entire calculation here, since it follows that of Sec.~\ref{mainCalcSC} almost identically. One obtains exactly the same estimate as for $\omega_{\text{q}} \theta^{\text{new}}$ in Eq.~\eqref{thetasFinal}:
\begin{equation}
   \omega_{\text{q}}\theta_L =\omega_{\text{q}}\theta_R = O(1)\left.\frac{\omega_{\text{q}}^2}{\omega_{\text{pl}} h_x \text{ln}\frac{\omega_{\text{pl}} }{h_x}} \right|_{s=1}, \quad \omega_{\text{q}}\theta^{\text{new}} = O(1)\left.\frac{\omega_{\text{q}}^2}{\omega_{\text{pl}} h_x \text{ln}\frac{\omega_{\text{pl}} }{h_x}} \right|_{s=1}  
\end{equation}
Indeed, the derivatives of $f$ turned out to be subleading in the derivation, and Hamiltonians with the same $b$-dependence will lead to the same bound.

Now recall that:
\begin{equation}
    H_{\text{CJJ}} = E_C \hat{n}^2  + E_J b \cos \hat{\phi} + E_L (\hat{\phi}- f)^2
\end{equation}
Applying the well approximation, we obtain again:
\bes
\label{HOequiv2}
\begin{align}
    H_{\text{CJJ,L}} = E_C \hat{n}^2 + E_J b \frac{ (\hat{\phi}+\pi)^2}{2},  \quad \phi \in [-\infty, \infty], \\
    H_{\text{CJJ,R}} = E_C \hat{n}^2 + E_J b \frac{ (\hat{\phi}-\pi)^2}{2},   \quad \phi \in [-\infty, \infty] ,
\end{align}
\ees
now with energy shifts $E_L (\pm \pi- f)^2$. The schedule for $f$ for CJJ is chosen in a way that results in the same energy shift. The derivative $-2E_L \hat{\phi}f'$ of the term $E_L (\hat{\phi}- f)^2$ is contains an operator diverging with the cutoff $\|\phi\| = \Theta(\sqrt{\Lambda})$. Since $\theta^{\text{JRS}}$ contains $\|H'\|^2= \Theta(\Lambda)$, it diverges while $\theta^{\text{new}} \sim \theta_L =\theta_R$ focusing on the behavior of the low-lying states in the well approximation is the same as for CSFQ:
\begin{equation}
    \theta^{\text{JRS}}_{\text{CJJ}} = \Theta(\Lambda), \quad   \omega_{\text{q}}\theta^{\text{new}}_{\text{CJJ}} = O(1)\left.\frac{\omega_{\text{q}}^2}{\omega_{\text{pl}} h_x \text{ln}\frac{\omega_{\text{pl}} }{h_x}} \right|_{s=1}  
\end{equation}
\section{Effective Hamiltonian}
\label{eHam}
In this section we will show that the effective evolution in a $d_P$-dimensional low-energy subspace that is an image of $P(s)$ is best described by a $d_P\times d_P$ effective Hamiltonian:
\beq
H_{\text{eff}}(s) = V(s)H(s)V^{\dag}(s) ,\quad V(s) = V_0 U^{\dag}_{\text{eff}}(s), \quad \frac{\partial}{\partial s}U_{\text{eff}}(s) =   [P',P]U_{\text{eff}}(s),  \quad U_{\text{eff}}(0)=I 
\eeq
where the isometry $V_0$ describes a freedom of choice of basis in the low-energy subspace at $s=0$.

Consider the equation of the approximate evolution $\Uad (s)|\phi_0\rangle = |\phi(s)\rangle$ generated by $\Had(s)$ [Eq.~\eqref{eq:Had}]:
\begin{equation}
    \frac{\partial}{\partial s}|\phi(s)\rangle = -i \Had(s)|\phi(s)\rangle . 
    \label{phiDef}
\end{equation}
This is written in the full Hilbert space even though we know that $\forall s>0, ~ P(s)|\phi(s)\rangle = \ket{\phi(s)}$ as long as the same holds for the initial state $|\phi_0\rangle$. 

This suggests that we could write the evolution as generated by a $d_P\times d_P$ matrix in the low-energy subspace - the effective Hamiltonian. Of course, one can trivially do this by first undoing the evolution generated by $\Uad$, i.e., by first changing the basis in a time-dependent manner via
\begin{equation}
   |\zeta(s)\rangle  = \Uad^\dag |\phi(s)\rangle  \quad \Longrightarrow \quad\frac{\partial}{\partial s}|\zeta(s)\rangle = \frac{\partial}{\partial s}|\phi_0\rangle = 0 .
\end{equation}
Let the eigenvectors of $H(0)$ in the low energy subspace be $\{|\lambda_i\rangle\}_{i=1}^{d_P}$, and let the basis vectors defining the new $d_P$-dimensional Hilbert space we map into be $\{|e_i\rangle\}_{i=1}^{d_P}$. Then the isometry $V_0$ corresponding to the projection $P_0\equiv P(0)$ can be chosen as:
\begin{equation} 
    V_0 = \sum_{i=1}^{d_P} |e_i\rangle\! \langle \lambda_i| .
\end{equation}
We use $V_0$ to form a ${d_P}$-dimensional Schr\"odinger equation
\begin{equation}
    |\psi(s)\rangle = V_0|\zeta(s)\rangle \quad \Longrightarrow \quad\frac{\partial}{\partial s}|\psi(s)\rangle =0 .
\end{equation}
Thus the effective ${d_P}\times {d_P}$ Hamiltonian governing the dynamics of $|\psi(s)\rangle$ is zero in this basis. The observables $O$ of the original system have to be transformed accordingly:
\begin{equation}
   O_V(s) = V_0\Uad ^\dag(s) O\Uad(s) V_0^\dag ,
\end{equation}
which is $t_f$-dependent. 

We would now like to present another time-dependent basis in which this $t_f$-dependence disappears. There are some additional reasons to consider a different effective Hamiltonian, to be discussed below. Define $U^G_{\text{eff}}(s)$ via:
\begin{equation}
      \frac{\partial}{\partial s}U^G_{\text{eff}}(s) =  (G+ [P',P])U^G_{\text{eff}}(s) ,
      \label{eq:Ueff}
\end{equation}
where $G=G(s)$ is a gauge (geometric connection) term in the generator for the basis change, which we assume to be block-diagonal ($G =PGP +QGQ$). We prove in Appendix~\ref{app:intertwining} that any such $U^G_{\text{eff}}$ will satisfy the intertwining property much like Eq.~\eqref{eq:intertwining} for $U_{\text{ad}}$:
\begin{equation}
    U^G_{\text{eff}}(s)P_0 = P(s)U^G_{\text{eff}}(s) .  
    \label{eq:Gint}
\end{equation}

We then let $U^G_{\text{eff}}$ be our time-dependent change of basis transformation:
\begin{equation}
    |\zeta(s)\rangle = U^{G\dag}_{\text{eff}}(s) |\phi(s)\rangle = U^{G\dag}_{\text{eff}}(s)\Uad (s)|\phi_0\rangle .
\eeq  
Now, $\frac{\partial}{\partial s} U^{G\dag}_{\text{eff}} =  U^{G\dag}_{\text{eff}}(G^\dag+[P,P'])$, so that, using Eq.~\eqref{phiDef}:
\begin{align}
    \frac{\partial}{\partial s}|\zeta\rangle = U^{G\dag}_{\text{eff}}(G^\dag+[P,P']-i\Had)\ket{\phi} = \\= U^{G\dag}_{\text{eff}}(G^\dag+[P,P']-i(t_f H+i[P',P]))U^G_{\text{eff}}|\zeta\rangle = U^{G\dag}_{\text{eff}}(G^\dag-i t_f H)U^G_{\text{eff}}|\zeta\rangle ,
\end{align}
where $H(s)$ is the full Hamiltonian in Eq.~\eqref{eq:Had}. Note that combining Eqs.~\eqref{eq:intertwining} and~\eqref{eq:Gint}, we see that $\ket{\zeta(s)}$  remains in the $s=0$ low-energy subspace: $\ket{\zeta(s)} =P_0\ket{\zeta(s)}$ for all $s$. Thus the isometry $V_0$ defined as before completes the mapping into the effective (e.g., qubit) $d_P$-dimensional Hilbert space:
\begin{align}
\label{eq:170a}
    |\psi(s)\rangle = V_0|\zeta(s)\rangle  
    \end{align}
Thus:
\beq
\frac{\partial}{\partial s}|\psi(s)\rangle =-i t_f H^G_{\text{eff}}(s) |\psi(s)\rangle  ,
\label{eq:172}
\end{equation}
where 
\beq
H^G_{\text{eff}}(s) = V^G(s)\left(H(s)+\frac{i}{t_f}G^\dag\right)V^{G\dag}(s) ,
\label{eq:HGeff}
\eeq
and we defined the time-dependent isometry 
\beq
V^G(s) = V_0 U^{G\dag}_{\text{eff}}(s) 
\eeq
into the effective basis at any $s$. Note that, combining our notation, we can write:
\begin{equation}
    |\psi(s)\rangle = V^G |\phi(s)\rangle, \quad |\phi(s)\rangle = V^{G\dag}|\psi(s)\rangle 
\end{equation}

The adiabatic theorem (Theorem~\ref{th:AT}) we have proven gives the bound [recall Eq.~\eqref{eq:ineq-th}]
\begin{equation}
    \|\ket{\phi(s)} - \ket{\phi_{\text{tot}}(s)}\| \leq b  = \theta/t_f \ ,
\end{equation}
where $\ket{\phi(s)}$ is the approximate evolution from Eq.~\eqref{phiDef}, while $\ket{\phi_{\text{tot}}(s)} = U_{\text{tot}}(s)\ket{\phi_0}$ is the true evolution generated by the Hamiltonian $H(s)$ in the full Hilbert space. Applying the expression for $|\phi(s)\rangle$ in terms of $|\psi(s)\rangle$, we get:
\begin{equation}
    \|V^{G\dag}\ket{\psi(s)} - \ket{\phi_{\text{tot}}(s)}\| \leq b  = \theta/t_f \ .
\end{equation}

This inequality means that $\ket{\psi(s)}$, the state evolving according to the effective Hamiltonian, after an isometry back to the total Hilbert space is close to the true state $\ket{\phi_{\text{tot}}(s)}$. Since $V_0V_0^\dag = I$ and since $V_0$ is an isometry (hence norm reducing), we have
\begin{align}
\| \ket{\psi(s)} - V^G(s)\ket{\phi_{\text{tot}}(s)} =
    \|V^G(s) (V^{G\dag}(s)\ket{\psi(s)} - \ket{\phi_{\text{tot}}(s)} ) \| \leq \\ \leq \|V^{G\dag}(s)\ket{\psi(s)} - \ket{\phi_{\text{tot}}(s)}\|\leq b  = \theta/t_f
    \ . \label{eq:173}
\end{align}

Let $u(s)$ be generated by $t_f H_{\text{eff}}(s)$, i.e., $\ket{\psi(s)} = u(s)\ket{\psi(0)}$ [Eq.~\eqref{eq:172}]. Note that $\ket{\phi(0)} =V^{G\dag}\ket{\psi(0)}$. We can rewrite Eq.~\eqref{eq:173} as:
\begin{equation}
    \forall \ket{\psi(0)}: \quad  \| (u(s) - V^G(s) U_{\text{tot}}(s)V^{G\dag}(s)\ket{\psi(0)}  \|\leq b  = \theta/t_f 
    \ .
\end{equation}
It follows immediately that the same bound holds for the evolution operators, as stated in the Introduction [recall Eq.~\eqref{eq:4}]:
\begin{equation}
    \|u(s) - V^G(s)U_{\text{tot}}(s)V^{G\dag}(s)\| \leq b\ .
\end{equation}

The observables of the original system transform as:
\begin{equation}
   O^G_{\text{eff}}(s) = V^G(s) OV^{G\dag}(s)\ .
   \label{eq:O_eff}
\end{equation}

In practice, $H_{\text{eff}}$ and $O_{\text{eff}}$ can be found by truncation of the total Hilbert space to some large cutoff, and working with truncated finite dimensional matrices $O,H, U,V$. The error introduced by the cutoff may be estimated by trying several cutoffs and extrapolating. 
We defer a more rigorous treatment of this error to future work. 

Let us now discuss the gauge $G$. There are two natural reasons for choosing $G=0$. The first is that 
if we wish to keep the basis change (and thus the operators $O_{\text{eff}}^G = V^G(s)OV^{G\dag}(s)$) $t_f$-independent, then $G$ itself must be $t_f$-independent. Thus, by Eq.~\eqref{eq:HGeff}, the only choice that leads to $t_f$-independent $H_{\text{eff}}^G(s)$ is $G =0$.

The second is that the choice $G=0$ is the one that minimizes the norm of the derivative of any observable. This can be interpreted as the desirable consequence of not imparting any additional geometric phases that artificially speed up the evolution of observables in the given observation frame. To show this explicitly, note first that since we assumed that $G$ is block-diagonal, we cannot choose the block-off-diagonal form $G = -[P',P]$ to cancel the time-dependence of the operators. 
Now, by Eq.~\eqref{eq:Ueff}:
\begin{equation}
     \frac{\partial}{\partial s}{O_{\text{eff}}^G} =  V^G [O, G + [P',P]]V^{G\dag} .
\end{equation}
When an operator $X$ is block-diagonal so that in particular $PXP =0$, then also $V^GXV^{G\dag}=0$ since $V^G$ just maps onto the space the projector selects. With this, it is clear that since $P[P,P']P=0$, we have:
\begin{equation}
     \|\frac{\partial}{\partial s}{O_{\text{eff}}^G}\| =  \|V^G [O, G]V^{G\dag}\| \geq 0 ,
\end{equation}
with the norm vanishing in general only when $G=0$.

\section{Conclusions}
\label{sec:conc}

Starting with Kato's work in the 1950's, work on the adiabatic theorem of quantum mechanics has resulted in rigorous bounds on the convergence between the actual evolution and the approximate, adiabatic evolution. These bounds were initially derived for Hamiltonians with bounded-norm derivatives, then conjectured without presenting the explicit form for the unbounded case, subject to assumptions restricting the class of Hamiltonians to being `admissible', which essentially meant that norms of certain functions of $H$ and its derivatives were not allowed to diverge. In this work we obtained new bounds which are presented in the explicit form, and can be applied after the introduction of an appropriate cutoff to Hamiltonians whose derivatives are unbounded. After the cutoff all the derivatives are bounded by a function of the cutoff scale, but our bounds capture the physically relevant cases where the adiabatic timescale is independent of the cutoff. To achieve this we introduced a different assumption, relating $H'$ to a power of $H$ via a simple-to-check positivity condition [Eq.~\eqref{eq:main-assump}]. With this assumption, we derived a new form of the adiabatic theorem. We expect this adiabatic theorem to prove to be useful in a variety of situations, e.g., in the context of adiabatic quantum computing using superconducting qubits or trapped ions, where the physical degrees of freedom correspond to (perturbed) harmonic oscillators. 

To demonstrate and illustrate the latter, we performed a calculation of the adiabatic timescale characterizing the accuracy of the qubit approximation of the circuit Hamiltonian of a capacitively shunted flux qubit. Specifically we considered a time evolution fashioned after quantum annealing that attempts to reduce the qubit transverse field $X$ linearly as $(1-s)X$. The result shows that after some $s^*$ close to $1$ the state generally escapes from the qubit approximation. Specifically, higher oscillator states become populated in each well. We do not expect this leakage effect to introduce a significant change in the outcome of a single-qubit quantum anneal, since the end-measurement is just a binary measurement of which well the flux is in, not the projection onto the lowest eigenstates. Thus, the non-qubit eigenstates become categorized as 0 or 1 depending on the sign of the flux. It remains an open question what the effect of this type of leakage is in the case of multi-qubit quantum dynamics, and whether it impacts the prospects of  a quantum speedup.

We thank Marius Lemm for insightful comments. This material is based upon work supported by the National Science Foundation the Quantum Leap Big Idea under Grant No.~OMA-1936388. 
Research was also sponsored by the Army Research Office and was
accomplished under Grant Number W911NF-20-1-0075. 
This research is also based upon work (partially) supported by the Office of
the Director of National Intelligence (ODNI), Intelligence Advanced
Research Projects Activity (IARPA) and the Defense Advanced Research Projects Agency (DARPA), via the U.S. Army Research Office contract W911NF-17-C-0050. The views and conclusions contained herein are those of the authors and should not be interpreted as necessarily
representing the official policies or endorsements, either expressed or
implied, of the ODNI, IARPA, DARPA, ARO, or the U.S. Government. The U.S. Government
is authorized to reproduce and distribute reprints for Governmental
purposes notwithstanding any copyright annotation thereon.

\appendix

\section{Time-dependent harmonic oscillator: a brute-force estimate \textit{vs.} the bound of Sec~\ref{EjEcSec}}
\label{app:A}

\begin{figure}
\centering
\includegraphics[width=0.5\columnwidth]{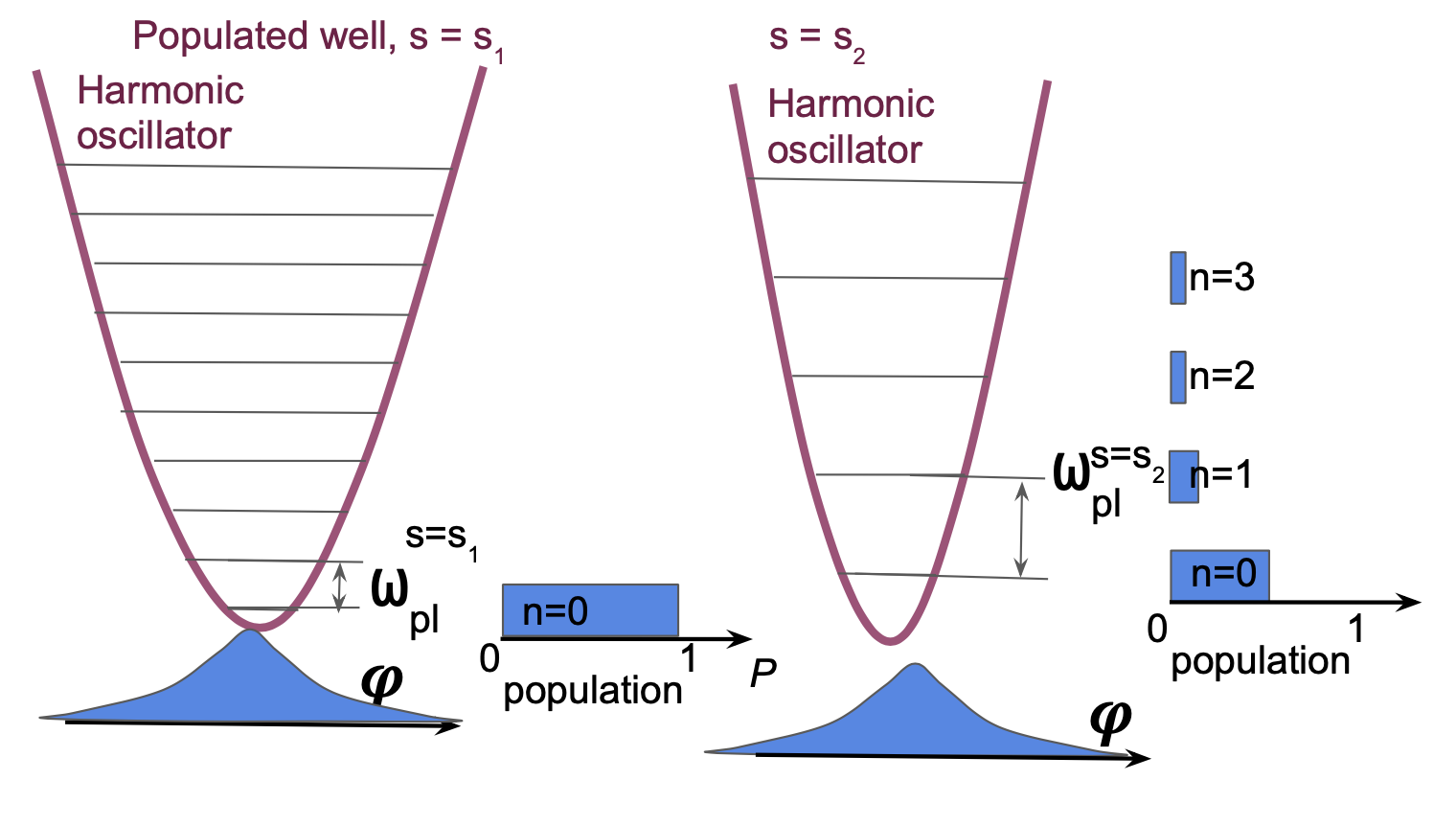}
\caption{For this figure, we assume that up to $s_1$ the evolution was fully adiabatic, but then the instantaneous approximation is applied to go from $s_1$ to $s_2$. The wavefunction is preserved, but since the eigenstates change it gets projected into the excited states. In the text, a more careful calculation of the leakage is carried out.}
\label{fig:physM}
\end{figure}

The well approximation of Eq.~\eqref{HOequiv} at every point $s$ along the anneal contains just harmonic potentials of different width, thus the leading order of leakage can be well described by changing the width of the harmonic potential by a dilation to $1/\sqrt{b(s)}$ of the $s=0$ width. We illustrate the leakage due to this effect in Fig.~\ref{fig:physM}. The diabatic evolution subject to the corresponding Hamiltonian
\begin{equation}
     H_{\text{0}} = E_C \hat{n}^2 + E_J b \frac{ \hat{\phi}^2}{2},  \quad \phi \in [-\infty, \infty],
\end{equation}
where we shifted the minimum to $\phi =0$, can be investigated in a brute-force manner, since we know the eigenstates $|m\rangle$ at every $s$, as well as their derivatives $|m'\rangle$. Indeed if we use the dilation operator: 
\begin{equation}
    U_d = e^{\frac{-i \text{ln} b}{4}(\hat{n}\hat{\phi} + \hat{\phi}\hat{n} )}, \quad U_d \psi(\phi) = b^{-1/4}\psi( \phi/\sqrt{b})
\end{equation}
we can express:
\begin{equation}
   |m_s\rangle = U_d |m_{s=0}\rangle, \quad |m_s'\rangle = U_d\frac{-ib'}{4b}(\hat{n}\hat{\phi} + \hat{\phi}\hat{n} ) |m_{s=0}\rangle . 
\end{equation}
We now write the time-dependent Schr\"odinger equation in the time-dependent eigenbasis, thus acquiring a geometric term:
\begin{align}
    &|\psi'(s)\rangle = -it_fH|\psi(s)\rangle, \quad |\psi(s)\rangle = \sum_m c_m(s)|m_s\rangle \\
    &\sum_m {c'_m}(s)|m_s\rangle + c_m(s)|{m_s}'\rangle = -\sum_mit_fHc_m(s)|m_s\rangle \\
     &{c'_k}(s) +\sum_m c_m(s)\langle k_s|{m_s}'\rangle = -\sum_mit_f c_m(s)\langle k_s|H|m_s\rangle \\
     &|c'(s)\rangle -\frac{ib'}{4b}(\hat{n}_{s=0}\hat{\phi}_{s=0} + \hat{\phi}_{s=0}\hat{n}_{s=0} )|c(s)\rangle = -it_f\omega_{\text{pl}}^s \hat{m}|c(s)\rangle ,
\end{align}
where $\hat{m}$ is just a diagonal matrix with $0,1,2\dots$ on the diagonal, and
\begin{equation}
    \hat{\phi}_{s=0} \sim \left(\frac{E_C}{E_J}\right)^{1/4} (a + a^\dag) ,\quad \hat{n}_{s=0}\sim \left(\frac{E_J}{E_C}\right)^{1/4} i(a - a^\dag) .
\end{equation}
and $a, a^\dag$ are 
the usual bosonic annihilation and creation operators. With this, we can estimate the leakage. Let us call 
\begin{equation}
    V(s) =\frac{-b'}{4t_fb}(\hat{n}_{s=0}\hat{\phi}_{s=0} + \hat{\phi}_{s=0}\hat{n}_{s=0} )
\end{equation}
a perturbation to the Hamiltonian. We split the interval $[0,t_f]$ into periods $2\pi/\omega_{\text{pl}}^s$. Over one period, we approximately consider $\omega_{\text{pl}}^s$ to be constant. We transform into the interaction picture: 
\begin{equation}
   V(t) =\frac{iO(1)b'}{t_fb}(a^2 e^{i2\omega_{\text{pl}}^s t} - a^{\dag2}e^{-i2\omega_{\text{pl}}^s t}) .
\end{equation}
We do not keep track of the numerical factors at this point. The leakage over one period is given by 
\begin{equation}
    |\delta c_s\rangle \approx i\int_0^{2\pi/\omega_{\text{pl}}^s} V(t) dt |0\rangle, \quad \| \delta c_s\| = O(1)\left| \left(\frac{\partial}{\partial t}\frac{b'}{2t_fb}\right)\int_0^{2\pi/\omega_{\text{pl}}^s} e^{-i2\omega_{\text{pl}}^s t} tdt \right|  = O(1)\frac{1}{t_f^2}\left(\text{ln} b\right)''\frac{1}{\omega_{\text{pl}}^{s2}} ,
\end{equation}
where the constant-in-$t$ term cancels in the rotating integral. Now what remains is to add contributions of all $s$ from intervals $2\pi/\omega_{\text{pl}}^s t_f$:
\begin{equation}
    \|\delta c\| = \int_0^1 ds (2\pi/\omega_{\text{pl}}^s t_f)^{-1}\| \delta c_s\| = \frac{O(1)}{t_f}\int_0^1 ds \left(\text{ln} b\right)''\frac{1}{\omega_{\text{pl}}^{s}} = \frac{O(1)}{t_f\omega_{\text{pl}}(0)}\int_0^1 ds \left(\text{ln} b\right)''\frac{1}{\sqrt{b}} .
\end{equation}
Taking the integral using Eqs.~\eqref{eq:b'} and~\eqref{changeOfVar}:
\begin{align}
   \int_0^1 ds \left(\text{ln} b\right)''\frac{1}{\sqrt{b}} &=  \left.\left(\text{ln} b\right)'\frac{1}{\sqrt{b}}\right|_0^1 - \int_0^1 ds \frac{b'}{b} {\left(\frac{1}{\sqrt{b}}\right)'} = \left.\frac{b'}{b^{3/2}}\right|_0^1  +\int_0^1 ds \frac{{b'}^2}{2b^{5/2}} \\
    &= \left.\frac{2}{\mathcal{A}(b)^{7/4}}  e^{\mathcal{A}(\sqrt{b} -1)} \right|_0^1  + \int_1^{\mathcal{B}} \frac{2}{\mathcal{A}^2b^{3}}  e^{2\mathcal{A}(\sqrt{b} -1)} \frac{\mathcal{A}(b)^{1/4}}{2}  e^{-\mathcal{A}(\sqrt{b} -1)} db  \\
  &  = \frac{2+o(1)}{\mathcal{A}\mathcal{B}^{7/4}}  e^{\mathcal{A}(\sqrt{\mathcal{B}} -1)}   +\frac{1+o(1)}{\mathcal{A}^2\mathcal{B}^{9/4}}  e^{\mathcal{A}(\sqrt{\mathcal{B}} -1)} .
\end{align}

The second term is subleading, thus
\begin{equation}
    E_C\theta = O(1) \frac{1}{\mathcal{A}^2\mathcal{B}^{7/4}}  e^{\mathcal{A}(\sqrt{\mathcal{B}} -1)} ,
\end{equation}
which exactly matches Eq.~\eqref{compare1} for $s^*=1, ~ b(s^*) = \mc{B}$. In other words, our brute-force calculation produces the same result as our bound.

\section{Proof of the intertwining relation, Eqs.~\eqref{eq:intertwining} and~\eqref{eq:Gint}}
\label{app:intertwining}

\begin{proof} 

It suffices to prove that $J(s)$ defined via
\begin{equation}
    J(s) \equiv U^G_{\text{eff}} (s) P_0 - P(s) U^G_{\text{eff}}(s)
    \label{eq:Jdef}
\end{equation}
vanishes for all $s$. Thus $J(s)$ is the ``integral of motion" of the differential equation satisfied by $\Uad (s)$. 

We can find the derivative using Eq.~\eqref{eq:Ueff}:
\bes
\begin{align}
    J' &= {U^G_{\text{eff}} }' P_0 - P'U^G_{\text{eff}}  -P{U^G_{\text{eff}} }' \\
    &=  G J +[P',P]U^G_{\text{eff}} P_0 - P'U^G_{\text{eff}}  -P[P',P]U^G_{\text{eff}} ,
\end{align}
\ees
where in the second equality we used $[P,G]=0$, which follows from $G$ being block-diagonal ($G = PGP +QGQ$).
Using the fact that $P'$ is block-off-diagonal [Eqs.~\eqref{eq:P'-offD} and~\eqref{eq:P'offdiag}], we simplify the last two terms as 
\beq
P'+P[P',P] = P'-PP' = QP' = QP'Q+QP'P = P'P ,
\eeq
where in the last equality we used $P'P = (PP'Q+QP'P)P = QP'P$.
Thus:
\bes
\begin{align}
    J'-G J &=  [P',P]U^G_{\text{eff}} P_0 - P'P U^G_{\text{eff}} \\
    &=[P',P]U^G_{\text{eff}} P_0 - (P'P-PP')P U^G_{\text{eff}}  = [P',P]J ,
\end{align}
\ees
i.e.,     
\beq
J' =(G +[P',P])J .
\eeq
Since $J(s)=0$ satisfies this equation and by definition of $J(s)$ [Eq.~\eqref{eq:Jdef}] we have $J(0)=0$, by uniqueness of the solution of a linear differential equation we obtain that $J(s)=0$ is the unique solution. This proves the desired property of $U^G_{\text{eff}}$. 

In the special case of $G(s) = -it_f H(s)$ we have $U_{\text{eff}}^G = U_{\text{ad}}$; thus proving Eq.~\eqref{eq:Gint} also proves Eq.~\eqref{eq:intertwining}.

\end{proof}


\bibliography{refs}
\end{document}